%% file: main.tex
\documentclass[11pt]{article}

\usepackage{fullpage}

\usepackage{amsfonts}
\usepackage{amsmath}
\usepackage{graphicx}
\usepackage{amssymb}
\usepackage{amsthm}
\usepackage{thmtools}
\usepackage{hyperref}
\usepackage{tikz}

\usepackage[inline]{enumitem}
\usepackage[numbers,square]{natbib}
\usepackage[ruled,vlined]{algorithm2e}

\usepackage[]{color-edits}% suppress

\addauthor{MS}{red}

\addauthor{MF}{blue}

\addauthor{SS}{magenta}

\newtheorem{proposition}{Proposition}[section]
\newtheorem{lemma}{Lemma}[section]
\newtheorem{theorem}{Theorem}[section]
\newtheorem{definition}{Definition}[section]

\newtheorem{maintheorem}{Main Theorem}
\newtheorem{observation}{Observation}[section]

\usepackage{cleveref}

\crefname{observation}{observation}{observations}
\newcounter{disccount}

\newcommand\blfootnote[1]{%
  \begingroup
  \renewcommand\thefootnote{}%
  \NoHyper\footnote{#1}\endNoHyper
  \addtocounter{footnote}{-1}%
  \endgroup
}

\newcommand{\contract}{\boldsymbol{\alpha}}
\newcommand{\eps}{\varepsilon}

\AddToHook{cmd/appendix/before}{%
    \crefalias{section}{appendix}%
}
\AddToHook{cmd/appendix/before}{%
    \crefalias{subsection}{appendix}%
}

\title{\textbf{Settling the Complexity Landscape of\\
Multi-Agent Contracts with Binary Actions}}

\author{
Michal Feldman\\
Tel Aviv University
\and
Maya Schlesinger\\
Tel Aviv University
\and
Shay Shani\\
Tel Aviv University
}

\date{}

\begin{document}

\maketitle
\blfootnote{
This project has been partially funded by the European Research Council (ERC) under the European Union's Horizon 2020 program (grant agreement No.~866132), by the European Union's Horizon Europe Program (grant agreement No.~101170373), by an Amazon Research Award, by the Israel Science Foundation Breakthrough Program (grant No.~2600/24), and by a grant from TAU Center for AI and Data Science (TAD), and by the NSF-BSF (grant number 2020788).
This research is supported by the Adams Fellowships Program of The Israel Academy of Sciences and Humanities.
}
\begin{abstract}
We study the computational complexity of optimal contract design in the multi-agent binary-action model, focusing on gross-substitutes reward functions and related classes. While additive rewards admit an FPTAS and general submodular rewards admit only constant-factor approximation, the complexity within the intermediate class of gross substitutes has remained largely open.

We uncover a fine-grained approximation landscape within this class. We first show that the optimal contract problem is APX-complete even for OXS rewards --- a strict subclass of gross substitutes rewards, ruling out a PTAS for gross substitutes unless $\mathsf{P}=\mathsf{NP}$. In contrast, for weighted matroid rank functions (WMRFs) --- another natural subclass of gross substitutes --- we obtain an EPTAS and show that no randomized FPTAS exists in general. We further identify the special case of partition (weighted) matroid rank functions, for which we obtain an FPTAS. 
This stands in contrast to the multi-agent {\em multi-action} setting, where no PTAS exists even for unweighted partition matroid rank functions.

Finally, we consider the broader class of ultra reward functions. While ultra rewards retain the tractability of gross substitutes in a related combinatorial contract model with a single agent, we show a sharp contrast in the multi-agent model: no polynomial-time randomized algorithm using value queries can achieve a $2^{o(n)}$-approximation in expectation. Together, our results reveal several qualitatively distinct computational regimes within and beyond gross substitutes, and identify submodularity as a crucial ingredient for the approximability of multi-agent contracts.
\end{abstract}

\input{intro2.tex}

\input{model.tex}

\input{no_ptas.tex}
\input{WMRF.tex}

\input{ultra.tex}

\bibliographystyle{alpha}
\bibliography{bib}

\input{appendix}

\end{document}

%% file: intro2.tex
\section{Introduction}
\label{sec:intro}
Multi-agent contracts constitute one of the major frontiers of algorithmic contract design and have attracted significant attention in recent years. A fundamental model in this line of work is the {\em multi-agent binary-action} model introduced by D\"utting et al.~\cite{DEFK23}, {building on the earlier work of~\cite{BFN2006}}. 
In this model, a principal delegates a project that can either succeed or fail to a set of $n$ agents, each of whom decides whether to exert effort or shirk.
Exerting effort incurs an agent-specific cost, and the probability of success is determined by a set function $f:2^{[n]}\rightarrow[0,1]$, where $f(S)$ denotes the probability of success when precisely the agents in $S$ exert effort. Upon success, the principal receives a reward of $1$, so $f(S)$ also represents her expected reward; accordingly, we refer to $f$ as the {\em reward} function. 

The principal cannot observe the agents' actions; she observes only whether the project succeeds or fails. Therefore, she incentivizes the agents by offering each agent a success-contingent payment, called a {\em contract}. Each contract induces a game among the agents. In this game, the principal's utility is her expected reward minus the expected payments to the agents, while each agent's utility is his expected payment minus the cost of exerting effort, if any. The principal's goal is to design the {\em optimal contract} --- a contract that maximizes her expected utility at a pure Nash equilibrium of the induced game.

A key quantity governing incentives is an agent's {\em marginal contribution}. For a set $S$ of agents exerting effort and an agent $i\in S$, the marginal contribution of $i$ is $f_i(S) := f(S)-f(S\setminus\{i\})$; 
namely, the increase in the probability of success due to agent $i$'s effort. If agent $i$ incurs cost $c_i$ for exerting effort, then the minimum success-contingent payment that incentivizes her to exert effort in $S$ is
$\frac{c_i}{f_i(S)}$.

The computational complexity of the optimal contract problem depends on the structure of the reward function $f$. Prior work has studied the problem through the hierarchy of complement-free set functions \cite{lehmann2001combinatorial}, establishing approximation guarantees for different classes of reward functions: for {\em additive} reward functions, the optimal contract problem admits an FPTAS, although computing an exact optimum is NP-hard~\cite{DEFK23}. In contrast, for {\em submodular} reward functions, a polynomial-time constant-factor approximation can be obtained~\cite{DEFK23}, but no PTAS exists~\cite{DEFK25}.
A natural question is whether prominent subclasses of submodular functions admit significantly better approximation guarantees.

A particularly well-studied class lying between additive and submodular functions is the class of {\em gross substitutes} (GS) functions. Gross substitutes were originally introduced to capture substitutability among goods~\cite{KelsoCrawford} and have since played a central role in combinatorial optimization, market design, and algorithmic game theory; see, e.g.,~\cite{gross_substitutability}. Their strong structural properties give rise to computational tractability in a variety of settings. 
Notably, this tractability also arises in the single-agent combinatorial contract model of D\"utting et al.~\cite{DEFK21}: the optimal contract can be computed in polynomial time for GS rewards, whereas the problem is NP-hard for the more general class of submodular rewards. 
{Recent work shows that this tractability extends beyond GS: \cite{FeldmanYashin2025} give a polynomial-time algorithm for the broader class of {\em ultra} rewards~\cite{lehmann2017ultra}. Ultra rewards contain GS but allow for certain forms of complementarity and are incomparable with submodular functions. In fact, for monotone valuations, GS is precisely the intersection of the ultra and submodular classes.}

Returning to the multi-agent binary-action model, since every GS function is submodular, the known constant-factor approximation for submodular rewards immediately applies. Whether the additional structure of GS can yield  stronger guarantees is a central open question, highlighted in~\cite{DuttingFT24}. In particular, does the optimal contract problem admit a PTAS for GS rewards? An EPTAS? Perhaps even an FPTAS? And what can be achieved for prominent subclasses of GS, such as weighted matroid rank functions (WMRFs) and OXS functions? 
Finally, can improved approximation guarantees be obtained beyond GS, for the broader class of ultra rewards? 
Notably, since ultra functions need not be submodular, even the constant-factor approximation known for submodular rewards does not automatically extend to this class.

In this work, we resolve these questions and characterize the approximation landscape for GS and ultra rewards, as well as important subclasses of GS. Our results reveal a surprisingly fine-grained picture.

\subsection{Our Results}
\label{sec:our-results}
\paragraph{Gross Substitutes (GS) rewards.}

Our main result resolves the open question of whether the optimal contract problem admits a PTAS for GS rewards. We show that the problem is APX-complete even for the natural subclass of OXS reward functions, and hence no PTAS exists for GS rewards unless $\mathrm{P}=\mathrm{NP}$.

\begin{maintheorem} [APX-hardness for {\bf OXS} rewards; \Cref{thm:apx-complete-oxs}]
The optimal contract problem for OXS reward functions (hence for GS rewards) is APX-complete.
\end{maintheorem}

{The hardness arises from the structure of agents' {\em marginal contributions}}. An OXS function admits a matching-based representation, and its value can therefore be computed efficiently using maximum-weight matching. The difficulty lies in the fact that marginal contributions, which determine the payments required to incentivize the agents, need not have a similarly simple structure. Removing a single agent can trigger a sequence of reassignments in the optimal matching, making its marginal contribution depend on a global rearrangement of the matching rather than only on its own allocation. Our reduction exploits precisely this phenomenon.

In contrast, for {\em weighted matroid rank functions} (WMRFs), another natural subclass of GS that is incomparable with OXS, we obtain an EPTAS and rule out a randomized FPTAS.
We prove these results in the value-oracle model; since value and demand oracles are polynomially equivalent for GS functions, the negative result holds even under demand-oracle access (see Section~\ref{sec:model} for definitions of value and demand oracles).

\begin{maintheorem} [EPTAS for {\bf WMRF} rewards, but no (randomized) FPTAS; \Cref{thm:wmrf-eptas,thm:wmrf-no-fptas}]

The optimal contract problem for weighted matroid rank (WMRF) reward functions admits an EPTAS, but no randomized  FPTAS.
\end{maintheorem}

The key structural observation underlying the positive result is that an optimal solution can always be chosen to be an independent set of the underlying matroid. On independent sets, the reward function becomes additive, and the marginal contribution of each selected agent coincides with its weight. 
Combining this with the characterization of optimal payments in terms of marginal contributions~\cite{DEFK23}, the payment required to incentivize an agent becomes independent of which other agents are incentivized.
This allows us to reduce the main optimization task to Budgeted Matroid Independent Set and leverage its known EPTAS~\cite{doronarad2023eptas}.

For the no-FPTAS result, we reduce from Exact Matroid Basis, which admits no randomized pseudo-polynomial-time algorithm in the independence-oracle model~\cite{doronarad2026lower}. We construct an instance in which the principal's utility over independent sets is a
quadratic function of their total weight, with a unique maximum at a prescribed weight. An independent set achieves this maximum if and only if it is a basis of the target weight. 
The discreteness of the attainable weights creates a gap between this maximum and the utility of every other independent set; an FPTAS could distinguish this gap, contradicting the oracle lower bound for Exact Matroid Basis.
Despite the negative result for general matroids, the problem admits an FPTAS for partition matroids.

\begin{maintheorem}[FPTAS for {\bf partition matroids};
\Cref{thm:wmrf-uniform-partition-fptas}]
If the matroid underlying the WMRF is a partition matroid,
then the optimal contract problem admits an FPTAS under value-oracle access.
\end{maintheorem}

Notably, the binary-action model was recently extended to a {\em multi-agent multi-action} setting, in which each agent controls multiple actions~\cite{DEFK25}. 
{In this more general setting, even unweighted partition matroid rank functions admit no PTAS~\cite{Equal-pay26}, in sharp contrast to the FPTAS we obtain here for the more general weighted case in the binary-action model.}

Taken together, our results reveal several distinct computational regimes within GS. Additive rewards admit an FPTAS; partition matroid rank functions retain this guarantee; general weighted matroid rank functions admit an EPTAS but no randomized FPTAS; and OXS rewards admit no PTAS. In particular, even among natural 
subclasses of GS, differences in structure lead to qualitatively different approximation guarantees.

\paragraph{Beyond Gross Substitutes: Ultra rewards.}

We next turn to ultra rewards and ask whether their structure suffices to obtain meaningful approximation guarantees in the multi-agent binary-action model. There is good reason to hope that this might be the case. As discussed above, in the single-agent combinatorial-action model, the polynomial-time algorithm for GS rewards extends to ultra rewards~\cite{FeldmanYashin2025}.
Thus, in the single-agent setting, tractability survives when submodularity is dropped but the ultra structure is retained.

The multi-agent binary-action setting exhibits a starkly different picture.
Specifically, we show that for ultra rewards, no polynomial-time randomized algorithm can achieve a $2^{o(n)}$-approximation in expectation.

\begin{maintheorem} [No $2^{o(n)}$-approximation under {\bf ultra} rewards; \Cref{thm:ultra-hardness}]
For ultra reward functions, 
no randomized algorithm making polynomially many value or demand queries can achieve a $2^{o(n)}$-approximation in expectation.
\end{maintheorem}

Thus, moving from GS to ultra leads to a sharp change in the approximation landscape: from constant-factor approximability for GS, and approximation schemes for important subclasses of GS, to subexponential inapproximability. 
Unlike in the single-agent setting, the ultra structure alone provides essentially no approximation guarantee in the multi-agent model, identifying {\em submodularity} as a crucial ingredient for approximability.

\paragraph{Submodular rewards.}
We finally revisit the constant-factor approximation for submodular rewards obtained by~\cite{DEFK23}, who give a $258$-approximation under value- and demand-oracle access and a $642.7$-approximation under value-oracle access alone. We improve these guarantees to $3.287$ and $6.128$, respectively.

Our approach builds on the square-root pricing framework of~\cite{DEFK23}, while introducing several refinements that improve the approximation factors. \cite{DEFK23}'s approach, roughly speaking, ``guesses'' different values for what the reward from the optimal solution is, and uses this guess to derive a nearly optimal solution. Our approach involves ``guessing'' the payment for the optimal solution as well, and adjusting the techniques and analysis accordingly. The full details are in \Cref{sec:gs-approximation} and \Cref{sec:submodular-value-oracle}.

\subsection{Related Work}
\label{subsec:related}
\paragraph{Algorithmic contract design.}
Our work belongs to the growing literature on algorithmic contract design, which studies the computational aspects of principal-agent problems. A central theme in this literature is the design of optimal or approximately optimal contracts when the underlying action spaces have combinatorial structure. We focus here on the works most closely related to ours, and refer the reader to recent surveys \cite{DuttingFT24,Feldman26,DuettingFT26} for a broader overview of algorithmic contract design.

\paragraph{Single-agent combinatorial contracts.}
The single-agent combinatorial contract model was introduced by D\"utting et al.~\cite{DEFK21}, where a single agent chooses a subset of costly actions and the reward is determined by a set function over the chosen actions. They showed that an optimal contract can be computed in polynomial time for GS rewards, while the problem is NP-hard for submodular rewards. {Subsequent work identified tractable
cases beyond gross substitutes. D\"utting, Feldman, and Gal-Tzur \cite{DFG24} gave a general procedure for enumerating the critical contract values and obtained a polynomial-time algorithm for supermodular rewards with submodular costs.} Feldman and Yashin \cite{FeldmanYashin2025} later extended the positive result to the broader class of ultra rewards. 
This was further generalized to a richer class of functions termed {\em All Substitutes and Complements} (ASC) by \cite{baldwin2026combinatorial}.
This stands in sharp contrast to our multi-agent setting, where we show that ultra rewards admit no subexponential approximation using polynomially many value queries.

\paragraph{Multi-agent contracts.}
The multi-agent binary-action model studied in this paper was introduced by D\"utting et al.~\cite{DEFK23}, building on an earlier multi-agent model of Babaioff, Feldman, and Nisan~\cite{BFN2006}. D\"utting et al.~\cite{DEFK23} initiated the systematic study of the computational complexity of the optimal contract problem in multi-agent settings across the complement-free hierarchy. For additive rewards, they gave an FPTAS; for submodular rewards, they obtained a constant-factor approximation using value oracles; and for XOS rewards, they obtained a constant-factor approximation using value and demand oracles.
They further established 
strong inapproximability for the more general class of subadditive rewards. 
{Ezra, Feldman, and Schlesinger~\cite{EFS24} subsequently ruled out a PTAS for submodular rewards under value-oracle access and showed that XOS rewards admit no constant-factor approximation using value queries
alone. D\"utting et al.~\cite{DEFK25} later strengthened the submodular hardness by ruling out a PTAS even with access to both value and demand oracles.}
  
{Recent work considers several extensions of multi-agent contract design.
Feldman et al.~\cite{FGPS25} study budget constraints in the binary-action model and show that the previously known approximation guarantees for the principal's utility extend asymptotically to this setting. D\"utting et al.~\cite{DEFK26Equilibria} study mixed, correlated, and coarse-correlated equilibria in the more general multi-agent combinatorial-actions model. They establish a constant-factor lifting theorem for XOS rewards and, for submodular rewards, a stronger robustness theorem. D\"utting et al.~\cite{DFGK26Online} study online arrivals in the binary-action model and obtain an $O(1)$-competitive policy for submodular rewards.}

Our work refines this landscape by studying important classes between additive and submodular rewards. In particular, we distinguish between two natural, incomparable subclasses of GS: weighted matroid rank functions, for which we obtain an EPTAS, and OXS functions, for which we rule out a PTAS. We also improve the constant-factor approximation for general GS rewards. 

\paragraph{Multi-agent multi-action contracts.}
The binary-action model was subsequently extended to the multi-agent multi-action setting by D\"utting et al.~\cite{DEFK25}, combining the multi-agent nature of the binary-action model with the combinatorial action space of the single-agent model. In this setting, each agent controls multiple actions, and the reward depends on the union of the actions selected by all agents. For submodular rewards, they obtain a constant-factor approximation with value and demand oracle access and show that no PTAS exists. More recent work establishes an FPTAS for additive rewards~\cite{oneactiontoomany} and rules out a PTAS for GS rewards, even for unweighted matroid rank functions corresponding to partition matroids~\cite{Equal-pay26}.
The latter hardness result does not apply to the binary-action setting
  studied in this paper, as its construction crucially exploits the richer
  action spaces available to individual agents. Our results show that, unless
  $\mathrm{P}=\mathrm{NP}$, no PTAS exists even in the binary-action model
  for GS rewards. At the same time, weighted matroid rank
  rewards admit an EPTAS but no randomized FPTAS in general, while an FPTAS
  exists when the underlying matroid is a partition matroid.

%% file: model.tex
\section{Model and Preliminaries}
\label{sec:model}
\paragraph{Model.}
We consider the multi-agent contract model studied in \cite{DEFK23}, wherein a
 principal seeks to induce effort from a set
$A=[n]$ of agents towards a delegated project. Each agent  $i\in A$ chooses whether to
exert effort, where exerting effort incurs a cost
$c_i\in\mathbb{R}_{\geq 0}$ and 
shirking
is costless. We identify
an action profile with the set $S\subseteq A$ of agents who exert effort. 
The project outcome is binary, $\omega\in\Omega=\{0,1\}$, where $0$ denotes failure
and $1$ denotes success. The outcome determines a reward to the principal;  the principal receives no reward upon failure, and
we normalize the reward from success to $1$. The probability of success is
determined by a function $f:2^A\rightarrow[0,1]$, where $f(S)$ is the
probability that the project succeeds under action profile $S$. Since
success gives a reward of 1, $f(S)$ is also the principal's expected reward. 
Throughout the paper, we assume that $f$ is normalized and
monotone, that is, $f(\emptyset)=0$ and $f(S)\leq f(T)$ for every
$S\subseteq T\subseteq A$. For every team of agents $S\subseteq A$ and agent $i\in S$, we denote the marginal contribution of agent $i$ to $S$ by $f_i(S):=f(S) -f(S\setminus\{i\})$.

\paragraph{Payments and incentives.} Since exerting effort comes with a cost, unless the principal acts the agent's preference is to do nothing, eventually leading to a project failure. Thus, to incentivize the agents to exert effort, the principal designs a payment scheme.
The agents' actions are hidden from the principal, so the principal cannot condition payments on the agents' effort. Instead, the principal conditions the payments on the project's outcome.
In the binary-outcome setting, it is
without loss of generality to consider success-contingent linear contracts.
A linear contract is represented by a vector
$\contract=(\alpha_i)_{i\in A}\in\mathbb{R}_{\geq 0}^{n}$, under which agent
$i$ receives $\alpha_i$ if the project succeeds and no payment if it fails. Agent $i$'s expected utility is defined as the expected payment he receives from the principal, minus his cost if he exerted effort. Thus, 
under contract $\contract$ and action profile $S$, agent $i$ obtains an expected utility of 
$\alpha_i f(S)-c_i$ if $i\in S$, and $\alpha_i f(S)$ otherwise. That way,  every
contract $\contract$ induces a game among the agents. We say that $\contract$
\emph{incentivizes} a set $S\subseteq A$ if the corresponding action
profile is a pure Nash equilibrium of the induced game among the agents, or equivalently, if no agent can improve
their utility by unilaterally changing their action. Thus,
\[
\begin{alignedat}{2}
    &\alpha_i f(S)-c_i \geq \alpha_i f(S\setminus\{i\})
    &\qquad& \text{for every } i\in S,\\
    &\alpha_i f(S) \geq \alpha_i f(S\cup\{i\})-c_i
    &\qquad& \text{for every } i\notin S.
\end{alignedat}
\]

\paragraph{The contract design problem. }
Under a contract $\contract$, the principal receives a net payoff
of $1-\sum_{i\in A}\alpha_i$ upon success and zero upon failure. Hence, the
principal's expected utility from a contract $\contract$ and action profile $S\subseteq A$ is
$(1-\sum_{i\in A}\alpha_i)f(S)$. As was observed in \cite{DEFK23}, for a fixed incentivizable set of agents $S\subseteq A$, this
quantity is maximized by setting $\alpha_i=0$ for every $i\notin S$ and
making the incentive constraint of every $i\in S$ tight, which gives
$\alpha_i=c_i/f_i(S)$. Consequently, the optimal principal's utility
obtainable from incentivizing $S$ is
\[
    g(S):=
    \left(
        1-\sum_{i\in S}
        \frac{c_i}{f_i(S)}
    \right)f(S).
\]
We interpret the ratio as $0$ when both $c_i$ and
$f_i(S)$ are zero, and as $+\infty$ when $c_i>0$ and
$f_i(S)=0$. Equivalently, $g(S)=-\infty$ whenever $S$
contains a positive-cost agent with zero marginal contribution.  The contract design problem is to find a contract $\contract$ and incentivized action profile $S$ that maximize the principal's utiltiy. By the above, it reduces to finding a set in 
$\arg\max_{S\subseteq A}g(S)$. We denote
an optimal set by $S^\star\in\arg\max_{S\subseteq A}g(S)$. 

\paragraph{The reward function $f$. }  
Since $f$ has an exponential representation size,
it is common to assume access to $f$ through one or both of two standard
oracles. A \emph{value oracle} receives a set $S\subseteq A$ and returns
$f(S)$. A \emph{demand oracle}, given a vector of nonnegative prices
$p=(p_i)_{i\in A}\in\mathbb{R}_{\geq0}^{n}$, returns a set
$T\in\arg\max_{S\subseteq A}\{f(S)-\sum_{i\in S}p_i\}$. 
Throughout the paper, we consider several well-studied classes of set functions, whose definitions we recall below. For simplicity of presentation, we slightly abuse notation and write $+i,-i$ for adding or removing agent $i$ from a set, respectively.

A set function $f:2^A\rightarrow[0,1]$ is:
\begin{itemize}
    \item \emph{Additive} if there exist values  $v_1,\dots ,v_n\in \mathbb{R}_{\ge 0}$, such that for every $S\subseteq A$, $f(S)=\sum_{i\in S}v_i$.
    \item \emph{Weighted matroid rank function} (WMRF) if there exists a matroid $\mathcal{M}=(A,\mathcal{I})$ and non-negative weights $w_1,\dots ,w_n \in \mathbb{R}_{\ge 0}$ such that for every $S\subseteq A$, 
    $f(S)=\max_{T\subseteq S, T\in \mathcal{I}}\sum_{i\in T}w_i$.
    \item \emph{OXS} if there exists a weighted bipartite graph
  $H=(A\cup L,E)$, 
  with nonnegative edge weights, such that, for every $S\subseteq A$,
  $
      f(S)
      =
      \max\left\{
          w(M):
          M\text{ is a matching in }H[S\cup L]
      \right\},
  $ where $H[S\cup L]$ is the graph induced by vertices $S\cup L$.
  When the graph $H$ and its edge weights are given as part of the input,
  we say that the OXS function is \emph{explicitly represented}.
    \item \emph{Gross-substitutes} (GS) if for every $X,Y\subseteq A$ and every
    $x\in X\setminus Y$, either
    $
    f(X)+f(Y)
    \leq
    f(X-x)+f(Y+x),
    $
    or there exists $y\in Y\setminus X$ such that
    $f(X)+f(Y)
    \leq
    f(X-x+y)+f(Y-y+x)$.
    \item \emph{Ultra} if for every $X,Y \subseteq A$ with $|X| \leq |Y|$ and every $x\in X\setminus Y$, there exists $y\in Y\setminus X$ such that 
    $
    f(X)+f(Y) \leq f(X-x+y)+f(Y-y+x)
    $.

\item \emph{Submodular} if for every $X\subseteq Y\subseteq A$ and every $i\in A\setminus Y$,
$
f_i(X+i)\geq f_i(Y+i).
$
\end{itemize}
It is well known that for monotone functions
$
\mathrm{GS} =  \mathrm{Submodular}\cap\mathrm{Ultra}
$~\cite{lehmann2017ultra} and,
\[
\mathrm{Additive} \subset \mathrm{WMRF}, \mathrm{OXS} \subset GS \subset \mathrm{Ultra,Submodular}\text{~\cite{lehmann2001combinatorial}}.
\]

\paragraph{A note about oracle models.} 
Demand queries are typically stronger than value queries. In particular, for any monotone function, a value query can be implemented using a polynomial number of demand queries, whereas the converse does not hold in general. For ultra functions, however, the converse also holds: a demand query can be implemented using polynomially many value queries~\cite{lehmann2017ultra,FeldmanYashin2025}. Consequently, the value- and demand-oracle models are polynomially equivalent for ultra and its subclasses.

\paragraph{Approximation terminology.}
Let $\mathrm{OPT}:=g(S^\star)$. For $\gamma\geq1$, a set $S\subseteq A$ is a $\gamma$-approximation if $\gamma g(S)\geq \mathrm{OPT}$. A poly-time algorithm is a \emph{constant-factor approximation} if it achieves a $\gamma$-approximation for some  constant $\gamma$, independent of the input size.

For $\eps\in(0,1)$, an
approximation scheme returns a set $S$ satisfying
$g(S)\geq(1-\eps)\mathrm{OPT}$. A
\emph{polynomial-time approximation scheme} ($\mathsf{PTAS}$) runs in polynomial time
for every fixed $\eps$, although the degree of the polynomial may depend
on $\eps$. An \emph{efficient polynomial-time approximation scheme}
($\mathsf{EPTAS}$) has running time of the form
$f(1/\eps)\cdot n^{O(1)}$, where $n$ denotes the input size and the
exponent of $n$ is independent of $\eps$. Finally, a
\emph{fully polynomial-time approximation scheme} ($\mathsf{FPTAS}$) runs in time
polynomial in both $n$ and $1/\eps$.
Viewed as classes of optimization problems, these notions satisfy
\[
    \mathsf{FPTAS}
    \subseteq
    \mathsf{EPTAS}
    \subseteq
    \mathsf{PTAS}
    \subseteq
    \mathsf{APX},
\]
where $\mathsf{APX}$ denotes the class of problems that admit a
polynomial-time constant-factor approximation.

%% file: no_ptas.tex
\section{No PTAS for Gross Substitutes Rewards}
\label{sec:no-ptas-gs}

In this section, we rule out a PTAS for the optimal contract problem with GS rewards. In fact, we establish the stronger result that the problem is $\mathsf{APX}$-complete even for the natural subclass of explicitly represented OXS reward functions.

\begin{theorem}
\label{thm:apx-complete-oxs}
The optimal contract problem with explicitly represented OXS rewards is $\mathsf{APX}$-complete under PTAS reductions. Consequently, unless $\mathsf{P}=\mathsf{NP}$, the optimal contract problem admits no PTAS for GS rewards, even when restricted to OXS rewards.
\end{theorem}
We reduce from Max-Cut on simple cubic graphs, which was shown to be
$\mathsf{APX}$-complete by \cite{alimonti2000apx}. For any such graph, we consider a propositional representation in the form of a CNF formula, where each vertex is represented by a variable, and its assignment represents which side of the cut it is given. In particular, for any edge $\{u,v\}$ in the graph, we add two clauses: $(x_u \lor \lnot x_v)$ and $(x_v \lor \lnot x_u)$. This representation crucially encodes the Max-Cut problem in the following way: the number of clauses satisfied by an assignment $\rho$ is exactly the total number of clauses minus the number of edges crossing the corresponding cut. 

With this propositional representation in mind, we present the construction of our reduced contract design instance (see \Cref{subsec:oxs-reduction}).  The set of agents includes two \emph{state agents} for each vertex, which encode the cut/assignment, and one \emph{detector agent} for each clause. In theory, the selected team of agents can be an arbitrary combination of state agents and detector agents, however, our construction is carefully crafted such that without loss of generality, an optimal set encodes a valid assignment/cut (see \Cref{subsec:assignment-sets}). More concretely, detector agents induce a large reward for the principal, and thus the principal seeks to incentivize all of them, at the lowest possible total payment. This is where the structure of the original instance comes in; the cut is encoded in the
detector marginals. Within any team of agents, a detector agent has a smaller marginal contribution when the corresponding state agents encode a valid assignment which satisfies its clause. Since the payment required by an agent is inversely proportional to its marginal
contribution, satisfied clauses are more expensive for the principal. Thus, the cheapest way to incentivize all detector agents is by minimizing the number of clauses satisfied by an assignment, and equivalently, maximizing the cut.
\subsection{From MAX-CUT to Contracts with OXS Rewards}
\label{subsec:oxs-reduction}
\paragraph{MAX-CUT on simple cubic graphs.}
Let $G=(V,E)$ be a simple cubic graph. Denote
$
    N:=|V|,
    |E|=\frac{3N}{2},
    m:=2|E|=3N.
$
For every edge $\{u,v\}\in E$, introduce the two clauses
\[
    C^1_{uv}:=(x_u\lor\neg x_v),
    \qquad
    C^2_{uv}:=(\neg x_u\lor x_v),
\]
and let $\mathcal C$ be the resulting collection of $m$ clauses.

An assignment $\rho:V\rightarrow\{0,1\}$ defines a cut by placing vertices
assigned $0$ and $1$ on opposite sides. If an edge crosses the cut, exactly
one of its two clauses is satisfied; otherwise, both are satisfied.
Therefore,
$
    \operatorname{sat}(\rho)
    =
    m-\operatorname{cut}(\rho),
$
where $\operatorname{sat}(\rho)$ and $\operatorname{cut}(\rho)$ denote the
number of satisfied clauses and crossing edges, respectively.

\paragraph{The OXS construction.}
We construct a weighted bipartite graph whose left side consists of agents
and whose right side consists of slots. Let 
$Z:=N+3m$.

For every vertex $u\in V$, introduce two state agents $u^0,u^1$ and a
vertex slot $s_u$. Selecting $u^\phi$ is intended to encode the assignment
$x_u=\phi$. For every clause $C\in\mathcal C$, introduce a detector agent
$d_C$ and a clause slot $s_C$. The edges are defined as follows:
\begin{enumerate}
    \item Each state agent $u^\phi$, $\phi\in\{0,1\}$, is connected to
    $s_u$ by an edge of weight $1/Z$.

    \item Each detector $d_C$ is connected to $s_C$ by an edge of weight
    $3/Z$.

    \item A state agent $u^\phi$ is connected to $s_C$ by an edge of
    weight $2/Z$ whenever the literal involving $x_u$ in $C$ is satisfied
    by setting $x_u=\phi$.
\end{enumerate}
Since $G$ is cubic, every state agent is connected
to exactly three clause slots, and the construction has polynomial size.

For every set of agents $S$, let $f(S)$ be the weight of a maximum matching
between the agents in $S$ and the slots. Since each vertex slot contributes
at most $1/Z$ and each clause slot contributes at most $3/Z$, it holds that 
$f(S) \leq \frac{N+3m}{Z} = 1$. 
Thus, $f$ is a normalized monotone OXS reward function. 
Fix $\eta:=1/100$, and set
\[
    c_{u^\phi}:=\frac{\eta}{100mZ},
    \qquad
    c_{d_C}:=\frac{\eta}{mZ}.
\]

For every incentivizable set $S$, 
write $t(S)
    :=
    \sum_{i\in S}
    \frac{c_i}{f_i(S)}$. 
    Then, $g(S)=\bigl(1-t(S)\bigr)f(S)$.
\Cref{fig:oxs-edge-gadget} illustrates the local OXS graph induced by a
single edge of the original cubic graph.
\input{oxs_gadget_figure}
\subsection{Assignment Sets}
\label{subsec:assignment-sets}
For every assignment $\rho:V\rightarrow\{0,1\}$ , we define the corresponding
\emph{assignment set} in the constructed contract instance
\[
    S_\rho
    :=
    \{u^{\rho(u)}:u\in V\}
    \cup
    \{d_C:C\in\mathcal C\}.
\]
Thus, $S_\rho$ contains exactly one state agent for every vertex and all
detector agents.
The following lemma shows that, 
among assignment sets, the principal's utility increases with the size of the corresponding cut.

\begin{restatable}[Utility of assignment sets]{lemma}{assignmentsetutilitylemma}
\label{lem:assignment-set-utility}
For every assignment $\rho:V\rightarrow\{0,1\}$, the corresponding
assignment set $S_\rho$ is incentivizable and satisfies
\begin{equation}
\label{eq:ass-set-utility}
    g(S_\rho)
    =
    1-\frac{\eta}{300}
    -\frac{\eta}{2}
    +\frac{\eta}{6m}\operatorname{cut}(\rho).
\end{equation}
\end{restatable}

The lemma rises from the natural matching for $S_\rho$, where every selected state agent is matched to its vertex slot and every detector is matched to its clause slot gives optimal matching value of 1. The marginal contributions of such matching happen to be $1/Z$ for state agents and $2/Z$ for detectors, if its clause is satisfied and $3/Z$ otherwise. Substituting the corresponding payment shares and using $\operatorname{sat}(\rho)=m-\operatorname{cut}(\rho)$ gives \eqref{eq:ass-set-utility}. For the formal proof see \Cref{app:ass-set-ut}.

The principal may choose an arbitrary set of agents, potentially containing both state agents of a vertex, omitting both state agents of another vertex, or omitting some detector agents. We next show that any such set can be replaced by an assignment set with weakly greater utility.

Adding or removing an agent may change the maximum matching and,
consequently, the marginal contributions and required payments of other
agents. We therefore compare an arbitrary incentivizable set directly with
an assignment set constructed from one of its maximum-weight matchings.

\begin{restatable}[Assignment-set lemma]{lemma}{assignmentsetlemma}
\label{lem:assignment-set}
For every incentivizable set $S$, one can compute in polynomial time an
assignment $\rho:V\rightarrow\{0,1\}$ such that $g(S_\rho)\geq g(S)$.
\end{restatable}

The proof (see \Cref{app:ass-set-wlog}) shows that the reward lost due to missing detector agents and unoccupied vertex slots is enough to cover the possible increase in the total payment share when passing from an arbitrary incentivizable set to an assignment set. The main issue is that this change may increase the payments of detector agents that were already selected. Cubicity ensures that each unoccupied vertex slot can be responsible for at most three such increases. Since $\eta$ is sufficiently small, the reward gained outweighs all possible payment increases. Therefore, the resulting assignment set has at least as much utility as the original set, showing that assignment sets are without loss of generality.

\subsection{Putting It All Together}
\label{subsec:oxs-apx-hardness}
We are now ready to prove \Cref{thm:apx-complete-oxs}. We first establish APX-hardness by showing that the construction above gives a PTAS reduction from Max-Cut on simple cubic graphs. We then complete the proof by showing that the problem belongs to $\mathsf{APX}$.

\begin{proof}[Proof of Theorem~\ref{thm:apx-complete-oxs}]
Let
$
    C^\star
    :=
    \max_{\rho:V\rightarrow\{0,1\}}
    \operatorname{cut}(\rho)
$
be the maximum cut value of the original graph. By \Cref{lem:assignment-set-utility,lem:assignment-set},
\[
    \mathrm{OPT}
    =
    1-\frac{\eta}{300}
    -\frac{\eta}{2}
    +\frac{\eta}{6m}C^\star.
\]

Fix $\eps\in(0,1)$ and set
$
    \eps_{\mathrm{OXS}}
    :=
    \frac{\eta\eps}{24}.
$
Suppose that $S$ is a $(1-\eps_{\mathrm{OXS}})$-approximate solution to
the induced contract instance. Since
Lemma~\ref{lem:assignment-set-utility} gives a positive-utility assignment
set, $\mathrm{OPT}>0$, and hence $S$ is incentivizable. By
Lemma~\ref{lem:assignment-set}, we can compute an assignment $\rho$ such
that $g(S_\rho)\geq g(S)$. Let $C:=\operatorname{cut}(\rho)$. Then
\[
    1-\frac{\eta}{300}-\frac{\eta}{2}+\frac{\eta}{6m}C
    =
    g(S_\rho)
    \geq
    g(S)
    \geq
    (1-\eps_{\mathrm{OXS}})\mathrm{OPT},
\]
where the first equality follows by \Cref{lem:assignment-set-utility}.
Rearranging and using $\mathrm{OPT}\leq1$ gives
$
    \frac{\eta}{6m}(C^\star-C)
    \leq
    \eps_{\mathrm{OXS}}\mathrm{OPT}
    \leq
    \eps_{\mathrm{OXS}},
    $ and hence $
    C^\star-C
    \leq
    \frac{6m\eps_{\mathrm{OXS}}}{\eta}=\frac{\eps m}{4}.
$

One can easily verify that
  $
      C^\star \geq \frac{m}{4}.
  $
Indeed, under a uniformly random partition, each edge crosses the cut with
probability \(1/2\), so some cut contains at least \(|E|/2\) edges. Since
  \(m=2|E|\), the claim follows.
Consequently,
$
    C^\star-C
    \leq
    \frac{\eps m}{4}
    \leq
    \eps C^\star,
    $ so $
    C\geq(1-\eps)C^\star.
$

Thus, a $(1-\eta\eps/24)$-approximate solution to the OXS contract instance
yields, in polynomial time, a $(1-\eps)$-approximate cut. Hence, this is a PTAS
reduction from Max-Cut on cubic graphs. Since cubic Max-Cut is APX-complete
\cite{alimonti2000apx}, the OXS contract problem is APX-hard under PTAS
reductions.

It remains to show that the problem belongs to $\mathsf{APX}$.
  Every OXS reward function is submodular, and the optimal contract problem
  with submodular reward functions admits a polynomial-time constant-factor
  approximation using value queries~\cite{DEFK23}. Given an explicit OXS
  representation, every value query can be answered in polynomial time by
  computing a maximum-weight bipartite matching. Hence, the optimal contract
  problem with explicitly represented OXS reward functions belongs to
  $\mathsf{APX}$.
Combining APX-hardness with membership in $\mathsf{APX}$ proves that the
  problem is $\mathsf{APX}$-complete under PTAS reductions. Since every OXS
  reward function is gross substitutes, a PTAS for GS rewards would also
  give a PTAS for explicitly represented OXS rewards. Thus, unless
  $\mathsf{P}=\mathsf{NP}$, no such PTAS exists.

\end{proof}

%% file: oxs_gadget_figure.tex
\begin{figure}[htbp]
    \centering
    \begin{tikzpicture}[
        x=1.15cm,
        y=0.95cm,
        state/.style={
            circle,
            draw=black,
            fill=black!3,
            minimum size=7.5mm,
            inner sep=1pt,
            font=\small
        },
        detector/.style={
            rectangle,
            rounded corners=2pt,
            draw=black,
            fill=black!8,
            minimum width=13mm,
            minimum height=7.5mm,
            inner sep=1pt,
            font=\small
        },
        slot/.style={
            rectangle,
            draw=black,
            fill=black!3,
            minimum width=13mm,
            minimum height=7.5mm,
            inner sep=1pt,
            font=\small
        },
        weight one/.style={draw=black!55, line width=0.5pt},
        weight two/.style={draw=blue!65!black, line width=1.15pt},
        weight three/.style={draw=orange!85!black, line width=2pt},
        continuation dot/.style={
            circle,
            draw=none,
            fill=blue!65!black,
            minimum size=1.5pt,
            inner sep=0pt
        }
    ]
        \begin{scope}[xshift=1cm]
        \node[font=\small\bfseries] at (0,5.55) {Agents};
        \node[font=\small\bfseries] at (6.2,5.55) {Slots};

        \node[state]    (u1) at (0, 4.0) {$u^1$};
        \node[detector] (d1) at (0, 3.0) {$d_{C^1_{uv}}$};
        \node[state]    (v0) at (0, 2.0) {$v^0$};

        \node[state]    (u0) at (0, 0.35) {$u^0$};
        \node[detector] (d2) at (0,-0.65) {$d_{C^2_{uv}}$};
        \node[state]    (v1) at (0,-1.65) {$v^1$};

        \node[slot] (su) at (6.2, 4.75) {$s_u$};
        \node[slot] (c1) at (6.2, 3.0) {$s_{C^1_{uv}}$};
        \node[slot] (c2) at (6.2,-0.65) {$s_{C^2_{uv}}$};
        \node[slot] (sv) at (6.2,-2.40) {$s_v$};

        \node[anchor=west, font=\scriptsize] at (6.8, 3.0)
            {$C^1_{uv}=(x_u\lor\neg x_v)$};
        \node[anchor=west, font=\scriptsize] at (6.8,-0.65)
            {$C^2_{uv}=(\neg x_u\lor x_v)$};
        \foreach \agent in {u1,v0,u0,v1} {
            \draw[weight two] (\agent.20) -- ++(0.49,0.13);
            \path (\agent.20) -- ++(0.95,0.26)
                node[pos=0.68, continuation dot] {}
                node[pos=0.82, continuation dot] {}
                node[pos=0.96, continuation dot] {};

            \draw[weight two] (\agent.-20) -- ++(0.49,-0.13);
            \path (\agent.-20) -- ++(0.95,-0.26)
                node[pos=0.68, continuation dot] {}
                node[pos=0.82, continuation dot] {}
                node[pos=0.96, continuation dot] {};
        }

        \draw[weight one] (u1) -- (su);
        \draw[weight one] (u0) to[out=18,in=205,looseness=1.08] (su);
        \draw[weight one] (v0) to[out=-18,in=155,looseness=1.08] (sv);
        \draw[weight one] (v1) -- (sv);

        \draw[weight two] (u1) -- (c1);
        \draw[weight two] (v0) -- (c1);
        \draw[weight two] (u0) -- (c2);
        \draw[weight two] (v1) -- (c2);

        \draw[weight three] (d1) -- (c1);
        \draw[weight three] (d2) -- (c2);
        \end{scope}

        \begin{scope}[yshift=-2.28cm]
        \node[anchor=west, inner sep=0pt, font=\normalsize\bfseries]
            at (-4.10,5.05)
            {Edge Weights};

        \draw[weight one] (-4.10,4.42) -- (-3.45,4.42);
        \node[anchor=west, font=\small] at (-3.25,4.42) {$1/Z$};

        \draw[weight two] (-4.10,3.74) -- (-3.45,3.74);
        \draw[weight two] (-4.10,3.36) -- ++(0.34,0);
        \path (-4.10,3.36) -- ++(0.65,0)
            node[pos=0.68, continuation dot] {}
            node[pos=0.82, continuation dot] {}
            node[pos=0.96, continuation dot] {};
        \node[anchor=west, font=\small] at (-3.25,3.55) {$2/Z$};

        \draw[weight three] (-4.10,2.65) -- (-3.45,2.65);
        \node[anchor=west, font=\small] at (-3.25,2.65) {$3/Z$};
        \end{scope}
    \end{tikzpicture}
    \caption{The local OXS graph induced by an edge $\{u,v\}$ of the
    original cubic graph.  A state agent is connected to a clause slot
    exactly when its state satisfies the corresponding literal.  A blue edge
    ending in three dots denotes a weight-$2/Z$ edge to a clause slot outside
    the displayed graph.  Each state agent has two such edges, induced by
    the other two graph edges incident to its vertex.}
    \label{fig:oxs-edge-gadget}
\end{figure}
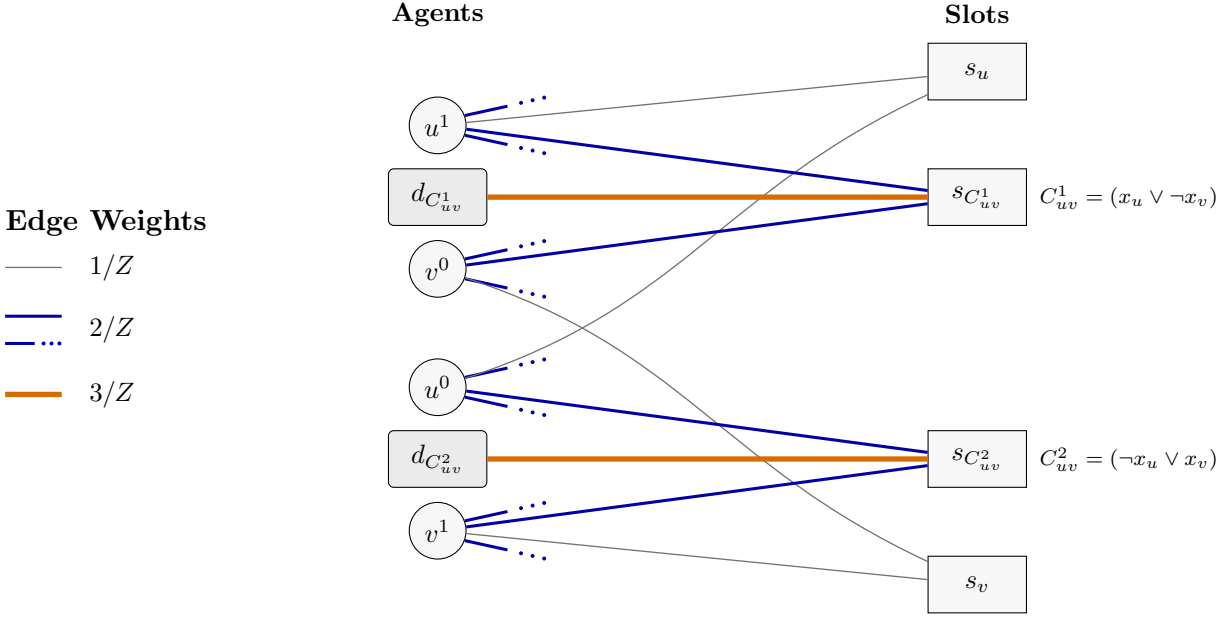

%% file: WMRF.tex
\section{Weighted Matroid Rank Function (WMRF) Rewards}
\label{sec:wmrf}

In this section, we study the optimal contract problem for weighted matroid rank function (WMRF) rewards.
In Section~\ref{subsec:wmrf-eptas}, we obtain an EPTAS for general WMRFs. In Section~\ref{subsec:wmrf-no-fptas}, we show that this guarantee cannot be strengthened to an FPTAS. Finally, in Section~\ref{subsec:wmrf-uniform-partition}, we show that, for the special case of partition matroids, the problem admits an FPTAS.

Before presenting our results, we introduce some notation, and establish a structural property of WMRFs that will be useful throughout the section.

Let $\mathcal{M}=(A,\mathcal{I})$ be the underlying matroid, 
$w_i\geq0$ be the weight of agent $i$, {and denote $w(S)=\sum_{i\in S} w_i$ for any $S\subseteq A$}. By definition,
$f(S) = \max_{T\subseteq S, T\in\mathcal{I}}\sum_{i\in T}w_i$.
Without loss of generality, we assume that
$w_i>0$ for every $i\in A$. Indeed, agents with zero marginal can be removed from every incentivizable set without decreasing the principal's utility.

The following lemma shows that there always exists an optimal solution in $\mathcal{I}$ (namely, an independent set). The proof appears in Appendix~\ref{app:wmrf-independent-proof}.

\begin{restatable}{lemma}{wmrfindependentoptimumlemma}
\label{lem:wmrf-independent-optimum}
For every set $S\subseteq A$, there exists an independent set
$T\subseteq S$ such that $f(T)=f(S)$ and $g(T)\geq g(S)$. In particular,
there exists an optimal solution in $\mathcal{I}$.
\end{restatable}

We note that value-oracle access allows us to implement an independence oracle. One can easily verify that, since all weights are positive, $S\in\mathcal{I}$ if and only if $f(S)=\sum_{i\in S}w_i$. 
Thus, after obtaining the weights from singleton queries, each independence query requires only one additional value query.

\subsection{An EPTAS for WMRF Rewards}
\label{subsec:wmrf-eptas}

The following theorem establishes an EPTAS for WMRF rewards.

\begin{theorem}[EPTAS for WMRF rewards]
\label{thm:wmrf-eptas}
For WMRF rewards, the optimal contract problem
admits an EPTAS
under value-oracle access.
\end{theorem}

Our algorithm iteratively solves the problem on a geometric grid of upper bounds on the total payment made to the agents. For any such upper bound, the problem reduces to an instance of Budgeted Matroid Independent Set (BMIS), which admits an EPTAS~\cite{doronarad2023eptas}. The key insight relates to bounding the size of the geometric grid, to ensure that we only solve a polynomial number of BMIS instances. To this end, we show
that either a singleton is already nearly optimal, or the optimal residual
share is bounded away from zero, ensuring that a polynomially-sized grid contains an
appropriate estimate.

Before proceeding, we introduce useful notations. For an agent $i\in A$, we denote by
$t_i:=\frac{c_i}{w_i}$ the required payment to incentivize $i$ to exert effort. Similarly, for a team $S\subseteq A$, we write $t(S)=\sum_{i\in S} t_i$ for the total payment required to incentivize $S$, and $\rho(S)=1-t(S)$ for the resulting principal's share of the reward. Thus, for every independent set $S\in\mathcal{I}$, $g(S) = \rho(S) w(S)$.

We next introduce the Budgeted Matroid Independent Set (BMIS) problem.

\begin{definition}[Budgeted Matroid Independent Set (BMIS)]
An instance of BMIS consists of a matroid $\mathcal{M}=(A,\mathcal{I})$ accessed
through an independence oracle, nonnegative profits $\{w_i\}_{i\in A}$,
nonnegative sizes $\{t_i\}_{i\in A}$, and a budget $B>0$. The objective is
  \[
      \max
      \left\{
          w(S):
          S\in\mathcal{I},\
          t(S)\leq B
      \right\}.
  \]
\end{definition}

\noindent We use the following known result.

\begin{theorem}[Theorem~1.1 in \cite{doronarad2023eptas}]

\label{thm:budgeted-matroid-eptas}
For every $\delta\in(0,1/2)$ and budget $B>0$, provided that
$t_i\in[0,B]$ for every item $i$, there is an algorithm that returns an independent set
$S_B$ satisfying $t(S_B)\leq B$ and 
    $w(S_B)
    \geq
    (1-\delta)
    \max_{S\in\mathcal{I}, t(S)\leq B}w(S)$. The running time of the algorithm is $h(1/\delta)\cdot L^{O(1)}$, 
where $L$ denotes the input size and $h$ is a computable function.
\end{theorem}

We write
$
    \mathsf{BMIS}(w,t,B,\delta)
$
for the set returned by the algorithm guaranteed in
Theorem~\ref{thm:budgeted-matroid-eptas}, using the independence oracle described above. 

\paragraph{{Bounding the optimal principal's share.}}
The following lemma ensures that the residual share of an optimal solution
is either bounded away from zero or a singleton is already nearly optimal, and in turn allows our polynomial grid size bound.
\begin{lemma}[Singleton or bounded residual]
\label{lem:wmrf-residual}
Let $S^\star\in\arg\max_{S\in\mathcal{I}}g(S)$, and suppose that
$g(S^\star)>0$. For every $\theta\in(0,1)$, at least one of the following
holds:
\begin{enumerate}
    \item There exists $i\in S^\star$ such that $g(\{i\})
        \geq
        (1-\theta)g(S^\star)$.

    \item The residual share satisfies $ \rho(S^\star)
        \geq
        \frac{\theta}{1+\theta}$.
\end{enumerate}
\end{lemma}

\begin{proof}
Suppose first that some $i\in S^\star$ satisfies
$w_i\geq(1-\theta)w(S^\star)$. Since $t_i\leq t(S^\star)$,
\[
    g(\{i\})
    =
    (1-t_i)w_i
    \geq
    \rho(S^\star)w_i
    \geq
    (1-\theta)g(S^\star).
\]

It remains to consider the case in which $w_i<(1-\theta)w(S^\star)$ for every $i\in S^\star$.
Fix $i\in S^\star$. Since $S^\star\setminus\{i\}$ is independent, optimality of $S^\star$
implies
\[
    \rho(S^\star)w(S^\star)
    \geq
    \bigl(\rho(S^\star)+t_i\bigr)
    \bigl(w(S^\star)-w_i\bigr).
\]
Rearranging and using
$w(S^\star)-w_i>\theta w(S^\star)$ gives $
    t_i
    <
    \frac{\rho(S^\star)w_i}
         {\theta w(S^\star)}$.
Summing over $i\in S^\star$, we obtain
\[
    1-\rho(S^\star)
    =
    t(S^\star)
    <
    \frac{\rho(S^\star)}{\theta},
\]
and therefore $\rho(S^\star)
    \geq
    \frac{\theta}{1+\theta}$, as needed.
\end{proof}
\paragraph{{Putting it all together.}}
We are now ready to prove \Cref{thm:wmrf-eptas}, by presenting \Cref{alg:wmrf-eptas}, our EPTAS algorithm, which 
considers the empty set and every singleton,
and then invokes the budgeted matroid procedure over a geometric grid of
residual shares.

\begin{algorithm}[t]
\DontPrintSemicolon
\caption{EPTAS for the optimal contract under WMRF $f$}
\label{alg:wmrf-eptas}

\KwIn{A WMRF $f:2^A\rightarrow[0,1]$ given by a value oracle,
costs $\{c_i\}_{i\in A}$, and
$\eps\in(0,1)$}

\KwOut{A set of agents $\widehat{S}$ to incentivize}

For every $i\in A$, let $w_i=f(\{i\})$ and discard every agent with
$w_i=0$. For every remaining $i\in A$, let $t_i=c_i/w_i$\;

Let
$\mathcal{C}:=\{\emptyset\}\cup\{\{i\}:i\in A\}$\;

Let
\[
    \theta:=\frac{\eps}{3},
    \qquad
    \rho_{\min}:=\frac{\theta}{1+\theta},
    \qquad
    K:=
    \left\lceil
        \log_{1/(1-\theta)}
        \left(\frac{1}{\rho_{\min}}\right)
    \right\rceil
\]

\For{$k=1,\ldots,K$}{
    Let $\rho_k:=(1-\theta)^k$ and $B_k:=1-\rho_k$\;

    Let $A_k:=\{i\in A:t_i\leq B_k\}$, and let
    $S_k\leftarrow\mathsf{BMIS}(w,t,B_k,\theta)$ on the restriction to
    $A_k$\;

    $\mathcal{C}\leftarrow\mathcal{C}\cup\{S_k\}$\;
}

\Return{
$\displaystyle
\widehat{S}\in
\arg\max_{S\in\mathcal{C}}
\bigl(1-t(S)\bigr)w(S)
$
}\;
\end{algorithm}

\begin{proof}[Proof of Theorem~\ref{thm:wmrf-eptas}]
By Lemma~\ref{lem:wmrf-independent-optimum}, let
$S^\star\in\mathcal{I}$ be an optimal solution. If
$g(S^\star)=0$, then the empty set is optimal, so assume that
$g(S^\star)>0$.

Apply Lemma~\ref{lem:wmrf-residual} with
$\theta=\eps/3$. If its first case holds, the corresponding singleton is
included in $\mathcal{C}$, and therefore
\[
    g(\widehat{S})
    \geq
    (1-\theta)g(S^\star)
    \geq
    (1-\eps)g(S^\star).
\]

Otherwise, if the second case holds, then
    $\rho(S^\star)\geq\rho_{\min}$.
Since $\rho_1=1-\theta$ and the grid decreases to a value at most
$\rho_{\min}$, there exists $k\in\{1,\ldots,K\}$ such that
$(1-\theta)\rho(S^\star)
    \leq
    \rho_k
    \leq
    \rho(S^\star)$.
The corresponding budget satisfies
\[
    B_k
    =
    1-\rho_k
    \geq
    1-\rho(S^\star)
    =
    t(S^\star),
\]
and, since all $t_i$ are nonnegative, $t_i\leq B_k$ for every
$i\in S^\star$. Thus, $S^\star\subseteq A_k$ is feasible for the restricted
budgeted matroid instance. By
Theorem~\ref{thm:budgeted-matroid-eptas},
\[
    w(S_k)\geq(1-\theta)w(S^\star),
    \qquad
    t(S_k)\leq B_k.
\]
The latter inequality implies
$1-t(S_k)\geq\rho_k$, and hence
\[
\begin{aligned}
    g(S_k)
    ~=~
    \bigl(1-t(S_k)\bigr)w(S_k) ~\geq~
    \rho_k(1-\theta)w(S^\star)
    ~\geq~ (1-\theta) \rho(S^\star) (1-\theta) f(S^\star) ~\ge~
    (1-\theta)^2g(S^\star).
\end{aligned}
\]
Since $\theta=\eps/3$, we have $(1-\theta)^2
    =
    1-\frac{2\eps}{3}+\frac{\eps^2}{9}
    \geq
    1-\eps$, and thus
    $g(S_k)\geq(1-\eps)g(S^\star)$.
Since $S_k\in\mathcal{C}$, the same guarantee holds for
$\widehat{S}$.

It remains to analyze the running time. By the definition of $K$,
\[
\begin{aligned}
    K
 =
    \left\lceil
        \frac{\log(1/\rho_{\min})}
             {-\log(1-\theta)}
    \right\rceil
    =
    O\left(
        \frac{1}{\theta}\log\frac{1}{\theta}
    \right)
    =
    O\left(
        \frac{1}{\eps}\log\frac{1}{\eps}
    \right),
\end{aligned}
\]
where we use
$\rho_{\min}=\theta/(1+\theta)$ and
$-\log(1-\theta)=\Theta(\theta)$.

Algorithm~\ref{alg:wmrf-eptas} makes $K$ calls to the Budgeted Matroid
Independent Set EPTAS, each with $\theta=\eps/3$.
Every independence query made by polynomial amount of call to
value queries as shown earlier. The total running time is therefore of the form
$h'(1/\eps)\cdot L^{O(1)}$, proving the theorem.
\end{proof}

\subsection{No FPTAS for WMRF Rewards}
\label{subsec:wmrf-no-fptas}

Theorem~\ref{thm:wmrf-eptas} establishes an EPTAS for WMRF rewards. The following theorem shows that an FPTAS is hopeless.

\begin{theorem}[No FPTAS for WMRF rewards]
\label{thm:wmrf-no-fptas}
The optimal contract problem for weighted matroid rank reward functions
admits no randomized FPTAS,
even in the demand-oracle model.
\end{theorem}

We reduce from Exact Matroid Basis, which admits no randomized
pseudo-polynomial-time algorithm in the independence-oracle
model~\cite{doronarad2026lower}. We construct an instance whose objective
over independent sets is a quadratic function with a unique maximum at a
prescribed integer weight. An FPTAS would distinguish whether this maximum
is attained, contradicting their lower bound. Value queries to the
constructed WMRF can be answered in polynomial time using the independence
oracle, so the reduction applies in the value-oracle model.

\paragraph{The Exact Matroid Basis problem.}
We first introduce the oracle version of the \emph{Exact Matroid Basis} problem.

\begin{definition}[Exact Matroid Basis]
\label{def:exact-matroid-basis}
An instance of \emph{Exact Matroid Basis} consists of a matroid
$\mathcal{M}=(E,\mathcal{I})$, an integer weight function
$\omega:E\rightarrow\mathbb{N}$, and a target $T\in\mathbb{N}$. For
$S\subseteq E$, let $\omega(S):=\sum_{i\in S}\omega(i)$. The problem asks
whether $\mathcal{M}$ has a basis $B$ satisfying $\omega(B)=T$.

In the oracle version of the problem, the ground set $E$, the weight
function $\omega$, and the target $T$ are given explicitly, whereas the
matroid is accessed through an independence oracle, which determines,
given a set $S\subseteq E$, whether $S\in\mathcal{I}$.
\end{definition}

Following \cite{doronarad2026lower}, a randomized algorithm for this
decision problem is required to return \textsc{Yes} with probability at
least $1/2$ on every yes-instance and to return \textsc{No} with
probability $1$ on every no-instance. The following result is
Theorem~1.3 in \cite{doronarad2026lower}.

\begin{theorem}[Oracle lower bound for Exact Matroid Basis]
\label{thm:oracle-emb}
There is no randomized algorithm for oracle Exact Matroid Basis whose
running time on every instance $(E,\mathcal{I},\omega,T)$ is
\[
    \bigl((|E|+1)(T+2)(\omega(E)+1)\bigr)^{O(1)}.
\]
\end{theorem}

In particular, oracle Exact Matroid Basis admits no randomized
pseudo-polynomial-time algorithm.

\paragraph{The reduction.}
Consider an instance $(E,\mathcal{I},\omega,T)$ of oracle Exact Matroid
Basis. An element $i\in E$ is a \emph{loop} if
$\{i\}\notin\mathcal{I}$. Loops can be identified using singleton
independence queries and belong to no basis, so deleting them preserves
the answer to the instance.

Let $r$ be the rank of the resulting matroid, which can be computed using
the standard matroid greedy algorithm. If $r=0$, the instance is a
yes-instance if and only if $T=0$, and if $T>\omega(E)$, it is a
no-instance. We therefore assume that $r\geq1$ and $T\leq\omega(E)$.

Set $H:=\omega(E)+1$ and $R:=T+rH$. Every basis $B$ has cardinality $r$,
and hence
\[
    \omega(B)+H|B|=\omega(B)+rH=R
    \quad\Longleftrightarrow\quad
    \omega(B)=T.
\]
On the other hand, every independent non-basis $S$ satisfies
\[
    \omega(S)+H|S|
    \leq \omega(E)+(r-1)H
    =rH-1
    <R.
\]
Thus, an independent set $S$ satisfies $\omega(S)+H|S|=R$ if and only if
it is a basis of weight $T$.

Using the same matroid, construct a WMRF with element weights
\[
    v_i:=\frac{\omega(i)+H}{2R}
\]
and agent costs $c_i:=v_i^2$. That is, define
\[
    f(S):=
    \max_{S'\subseteq S, S'\in\mathcal{I}}
    \sum_{i\in S'}v_i.
\]
Since all loops were removed, $v_i=f(\{i\})>0$ for every $i\in E$.

The reward function is normalized. Indeed, for every independent set
$S\in\mathcal{I}$,
\[
\begin{aligned}
    f(S)
    &=\sum_{i\in S}v_i
     =\frac{\omega(S)+H|S|}{2R}\\
    &\leq\frac{\omega(E)+rH}{2R}
     =\frac{(r+1)H-1}{2(T+rH)}
     <1,
\end{aligned}
\]
and since the value of an arbitrary set is the value of a maximum-weight
independent subset, $f$ takes values in $[0,1]$.

A value query on a set $S\subseteq E$ asks for a maximum-weight
independent subset of $S$ with respect to the weights $(v_i)_{i\in E}$.
It can therefore be answered exactly by sorting the elements of $S$ by
non-increasing weight and applying the weighted matroid greedy algorithm
using the independence oracle. Thus, every value query to the constructed
WMRF can be simulated in polynomial time.

Let
\[
    \mathrm{OPT}:=\max_{S\subseteq E}g(S)
\]
denote the optimal principal's utility in the constructed instance.

\begin{lemma}[Gap of the reduction]
\label{lem:wmrf-emb-gap}
The constructed contract instance satisfies the following:
\begin{enumerate}
    \item If the Exact Matroid Basis instance is a yes-instance, then
    $\mathrm{OPT}=1/4$.
    \item If the Exact Matroid Basis instance is a no-instance, then
    \[
        \mathrm{OPT}\leq\frac14-\frac{1}{4R^2}.
    \]
\end{enumerate}
\end{lemma}

\begin{proof}
By \Cref{lem:wmrf-independent-optimum}, every set can be replaced by an
independent subset of weakly greater utility. It therefore suffices to
analyze independent sets.

Fix $S\in\mathcal{I}$. The marginal contribution of every agent
$i\in S$ is $v_i$, and $c_i/v_i=v_i$. Consequently,
\[
\begin{aligned}
    g(S)
    &=\left(1-\sum_{i\in S}v_i\right)\sum_{i\in S}v_i
     =\left(1-\frac{\omega(S)+H|S|}{2R}\right)
       \frac{\omega(S)+H|S|}{2R}\\
    &=\frac14-
      \frac{\bigl(\omega(S)+H|S|-R\bigr)^2}{4R^2}.
\end{aligned}
\]

Suppose first that the Exact Matroid Basis instance is a yes-instance.
Then there exists a basis $B$ satisfying
$\omega(B)+H|B|=R$, so the squared term vanishes and $g(B)=1/4$. Since
the expression above is at most $1/4$ for every independent set, it
follows that $\mathrm{OPT}=1/4$.

Suppose now that the Exact Matroid Basis instance is a no-instance. Then
$\omega(S)+H|S|\neq R$ for every independent set $S$. Since both
quantities are integers,
\[
    \left|\omega(S)+H|S|-R\right|\geq1,
\]
and hence
\[
    g(S)\leq\frac14-\frac{1}{4R^2}.
\]
By \Cref{lem:wmrf-independent-optimum}, the same upper bound holds for
every set, and therefore for $\mathrm{OPT}$.
\end{proof}

We are now ready to prove \Cref{thm:wmrf-no-fptas}.

\begin{proof}[Proof of Theorem~\ref{thm:wmrf-no-fptas}]
Suppose toward contradiction that there exists a randomized FPTAS
$\mathcal{A}$ for the optimal contract problem with WMRF rewards under
value-oracle access.

Given an instance $(E,\mathcal{I},\omega,T)$ of oracle Exact Matroid
Basis, handle the directly decidable cases above and otherwise construct
the corresponding contract instance. Run $\mathcal{A}$ with 
$
    \eps:=\frac{1}{2R^2},
$
and let $\widehat{S}$ be the returned set.

If the Exact Matroid Basis instance is a yes-instance, then the
approximation guarantee in expectation gives
\[
    \mathbb{E}\bigl[g(\widehat{S})\bigr]
    \geq
    (1-\eps)\mathrm{OPT}
    =
    \left(1-\frac{1}{2R^2}\right)\frac14
    =
    \frac14-\frac{1}{8R^2}.
\]
Let
\[
    \tau:=\frac14-\frac{1}{4R^2},
    \qquad
    p:=\Pr\bigl[g(\widehat{S})>\tau\bigr].
\]
Since $\mathrm{OPT}=1/4$ in a yes-instance, every output satisfies
$g(\widehat{S})\leq1/4$. Therefore, 
$
    \mathbb{E}\bigl[g(\widehat{S})\bigr]
    \leq
    p\cdot\frac14+(1-p)\tau=
    \frac14-\frac{1-p}{4R^2}.
$
Combining the two bounds gives $p\geq1/2$.

If the Exact Matroid Basis instance is a no-instance, then every output
satisfies
\[
    g(\widehat{S})
    \leq
    \mathrm{OPT}
    \leq
    \frac14-\frac{1}{4R^2}
    =
    \tau.
\]
The two cases can therefore be distinguished by returning \textsc{Yes} if
and only if
$
    g(\widehat{S})>\tau.
$

The value of $g(\widehat{S})$ can be computed exactly using a value query
for $\widehat{S}$ and one for every
$\widehat{S}\setminus\{i\}$, $i\in\widehat{S}$. Thus, on a no-instance
the resulting algorithm always returns \textsc{No}, while on a
yes-instance it returns \textsc{Yes} with probability at least $1/2$.

The constructed contract instance can be computed in polynomial time.
Moreover,
\[
    \frac{1}{\eps}=2R^2,
    \qquad
    R=T+r\bigl(\omega(E)+1\bigr)
      \leq T+|E|\bigl(\omega(E)+1\bigr).
\]
and every value query made by $\mathcal{A}$ can be simulated using
polynomially many independence queries. Thus, the resulting algorithm for
oracle Exact Matroid Basis has running time
\[
    \bigl((|E|+1)(T+2)(\omega(E)+1)\bigr)^{O(1)},
\]
contradicting \Cref{thm:oracle-emb}.
\end{proof}

\subsection{FPTAS for WMRF Rewards with Partition Matroids}
\label{subsec:wmrf-uniform-partition}
 Recall that a partition matroid is specified by a partition
  $
      A=A_1\mathbin{\dot\cup}\cdots\mathbin{\dot\cup}A_h
  $
  into blocks, together with capacities $r_1,\ldots,r_h$. A set
  $S\subseteq A$ is independent if and only if
  $
      |S\cap A_\ell|\leq r_\ell
      $ for every $\ell\in[h].
  $
\begin{theorem}[FPTAS for partition matroids]
\label{thm:wmrf-uniform-partition-fptas}
The optimal contract problem for WMRF rewards 
corresponding to partition matroids admits an FPTAS, under value-oracle access.
\end{theorem}

\begin{proof}[Proof of \Cref{thm:wmrf-uniform-partition-fptas}]
We first recover an equivalent partition representation. With independence oracles we compute a basis $B$. For every
$i\notin B$, let
\[
    X_i
    :=
    \left\{
        j\in B:
        (B\setminus\{j\})\cup\{i\}\in\mathcal I
    \right\}.
\]
If $i$ belongs to a block $A_\ell$, then $X_i=B\cap A_\ell$. Hence, for
each distinct set $X$ among the sets $X_i$, form the block
\[
    X\cup\{i\in A\setminus B:X_i=X\}
\]
with capacity $|X|$. Place every basis element that occurs in none of the
sets $X_i$ in a singleton block of capacity one. Indeed, every block
containing a nonbasis element is recovered in this way, whereas a block
containing only basis elements has capacity at least its size and therefore imposes no
restriction. Thus, the resulting partition matroid has exactly the same
independent sets. The construction uses $O(n^2)$ independence queries, and
hence $O(n^2)$ value queries.

Write the recovered representation as
\[
    A=A_1\mathbin{\dot\cup}\cdots\mathbin{\dot\cup}A_h,
    \qquad
    |S\cap A_\ell|\leq r_\ell \quad(\ell\in[h]).
\]
We may assume that $r_\ell\leq |A_\ell|$ and write
$r:=\sum_{\ell=1}^h r_\ell$ for the rank. The case $r=0$ is immediate, so
assume that $r\geq1$.

Guess the largest weight in an optimal solution. For each
$b\in\{w_i:i\in A\}$, discard agents with $w_i>b$ and set
\[
    \overline w_i:=
    \left\lfloor\frac{r w_i}{\eps b}\right\rfloor,
    \qquad
    \overline w(S):=\sum_{i\in S}\overline w_i.
\]
Every remaining agent has rounded weight at most
$\lfloor r/\eps\rfloor$, and every independent set has at most $r$ agents.
Thus, only rounded weights
\[
    x\in\left\{0,\ldots,r\left\lfloor\frac r\eps\right\rfloor\right\}
\]
need to be considered.

Fix a guess $b$, and list the remaining agents of block $A_\ell$ as
$i_{\ell,1},\ldots,i_{\ell,m_\ell}$. We process the blocks in order. Let
$D_{\ell,j}(k,x)$ be the minimum value of $t(S)$ after all earlier blocks and
the first $j$ agents of $A_\ell$ have been processed, where $k$ agents have
been selected from the current block and $\overline w(S)=x$. Here
$0\leq k\leq r_\ell$, and $x$ ranges over the values specified above.

Initially, $D_{1,0}(0,0)=0$, and all other entries of $D_{1,0}$ are
$+\infty$. When $\ell>1$, initialize the next block by
\[
    D_{\ell,0}(0,x)
    =
    \min_{0\leq k'\leq r_{\ell-1}}
        D_{\ell-1,m_{\ell-1}}(k',x),
    \qquad
    D_{\ell,0}(k,x)=+\infty \quad(k>0).
\]
For $j=1,\ldots,m_\ell$, letting $i=i_{\ell,j}$, use
\[
    D_{\ell,j}(k,x)
    =
    \min\left\{
        D_{\ell,j-1}(k,x),
        D_{\ell,j-1}(k-1,x-\overline w_i)+t_i
    \right\},
\]
where entries with indices outside their ranges have value $+\infty$.
By induction over the processed agents and blocks, for every $x$,
\[
    \min_{0\leq k\leq r_h}D_{h,m_h}(k,x)
\]
is exactly the minimum value of $t(S)$ over all independent subsets $S$ of
the agents retained for the current guess that satisfy $\overline w(S)=x$.
A corresponding set is recovered using back pointers. The algorithm runs
this dynamic program for every guess $b$ and returns the recovered set of
maximum true utility, also considering the empty set.

By \Cref{lem:wmrf-independent-optimum}, fix an optimal independent set
$S^\star$. If $g(S^\star)=0$, the empty set is optimal. Otherwise, consider
the guess $b=\max_{i\in S^\star}w_i$. No element of $S^\star$ is discarded.
For $x=\overline w(S^\star)$, let $T$ be the set recovered from the
corresponding final entry. The definition of the table gives
$t(T)\leq t(S^\star)$, and therefore $\rho(T)\geq\rho(S^\star)$. Moreover,
the definition of $\overline w_i$ gives, for every retained agent $i$,
\[
    \frac{\eps b}{r}\overline w_i\leq w_i,
    \qquad
    \frac{\eps b}{r}\overline w_i
    \geq w_i-\frac{\eps b}{r}.
\]
Consequently, 
\[
    w(T)
    \geq \frac{\eps b}{r}\overline w(T)
     = \frac{\eps b}{r}\overline w(S^\star)\geq w(S^\star)-|S^\star|\frac{\eps b}{r}
     \geq w(S^\star)-\eps b
     \geq(1-\eps)w(S^\star).
\]
Here the last inequality uses $b\leq w(S^\star)$.
It follows that
\[
    g(T)=\rho(T)w(T)
    \geq(1-\eps)\rho(S^\star)w(S^\star)
    =(1-\eps)g(S^\star).
\]
There are $O(r^2/\eps)$ possible values of $x$. For a fixed guess, the
number of table entries is
\[
    O\left(
        \frac{r^2}{\eps}
        \left(r+\sum_{\ell=1}^h m_\ell r_\ell\right)
    \right)
    =
    O\left(\frac{nr^3}{\eps}\right),
\]
and the transitions and block initializations fit within the same bound.
Since there are at most $n$ guesses, the total running time is
$O(n^2r^3/\eps)$, including backtracking and evaluating the recovered sets.
\end{proof}

%% file: ultra.tex
\section{Ultra Rewards}
\label{sec:ultra}
In this section we show that the optimal contract problem for ultra rewards admits no subexponential approximation using polynomially many value queries.

\begin{theorem}[No subexponential approximation for ultra rewards]
\label{thm:ultra-hardness}
For ultra reward functions, no randomized algorithm that makes a
polynomial number of value and demand queries can achieve a
$2^{o(n)}$-approximation in expectation to the optimal contract problem.
\end{theorem}

We construct a family of ultra reward functions, indexed by a hidden set
$T\subseteq A$ of cardinality $\lfloor n/2\rfloor$. The value of every set other than the hidden set depends only on its cardinality, while the value at the hidden set is slightly lower than its size suggests. We show that this single downward perturbation preserves the ultra structure. We then assign equal costs to the agents such that only sets of the form $T \cup \{x\}$ for some $x\in A\setminus T$ yield a positive principal's utility. Since we only perturb a single value of $f$, value queries reveal negligible information on the hidden set $T$. Thus, since there are exponentially many possible hidden sets, when selecting the hidden set uniformly at random, Yao's principle gives the result.
\subsection{The Hard Family} 
\label{subsec:ultra-hard-family}
Fix $n\geq4$ and let $A=[n]$. Set 
$ 
r:=\left\lfloor\frac{n}{2}\right\rfloor, k:=r+1, \eps:=\frac{1}{4k}. 
$ 
Define the cardinality profile 
\[ 
\phi(j) := \eps\min\{j,r\} + \frac12\mathbf{1}[j\geq k], \qquad j\in\{0,\ldots,n\}. 
\] 
For every set $T\subseteq A$ with $|T|=r$, define the set function $f_T$ by 
\begin{equation} 
\label{eq:ultra-hard-function} f_T(S) := \phi(|S|) - \eps\mathbf{1}[S=T]. 
\end{equation} 
For $i\in S$, write $(f_T)_i(S):=f_T(S)-f_T(S\setminus\{i\})$ for the marginal contribution of $i$ under $f_T$.
Note that $\phi(|S|)$ depends only on the cardinality of $S$. Every such valuation is ultra by~\cite{lehmann2017ultra}. The following lemma handles the single exceptional hidden-set $T$. \begin{lemma}
\label{lem:ultra-single-dent}
For every $T\subseteq A$ and every $\eta\geq0$, the set function $ h_T(S):=-\eta\mathbf{1}[S=T] $ is ultra.
\end{lemma}

\begin{proof}
Fix $X,Y\subseteq A$ with $|X|\leq|Y|$, and let $x\in X\setminus Y$. For every $y\in Y\setminus X$, write 
\[
X_y:=X-x+y, \qquad Y_y:=Y-y+x. 
\]
Since $x\notin X_y$ and $x\in Y_y$, the two exchanged sets cannot both equal $T$. Suppose first that $X=T$ or $Y=T$. Since $X\neq Y$, exactly one of the two sets equals $T$, and hence \[
h_T(X)+h_T(Y)=-\eta. 
\]
For every $y\in Y\setminus X$, at most one of $X_y$ and $Y_y$ equals $T$, so 
\[
h_T(X_y)+h_T(Y_y)\geq-\eta. 
\]
Thus, every choice of $y$ satisfies the exchange inequality. Suppose now that neither $X$ nor $Y$ equals $T$, so that $h_T(X)+h_T(Y)=0$. If $|Y\setminus X|=1$, then $|X\setminus Y|=1$ as well. Hence, for the unique $y\in Y\setminus X$, the exchange swaps the two sets: 
\[
X_y=Y, \qquad Y_y=X. 
\]
The right-hand side is therefore also zero. Finally, suppose that $|Y\setminus X|\geq2$. If $x\in T$, then $X_y\neq T$ for every $y$, and at most one choice of $y$ can make $Y_y=T$. If $x\notin T$, then $Y_y\neq T$ for every $y$, and at most one choice of $y$ can make $X_y=T$. Since at least two choices are available, there exists $y\in Y\setminus X$ for which neither exchanged set equals $T$. For this choice, 
\[ 
h_T(X_y)+h_T(Y_y)=0, 
\]
which completes the proof. 
\end{proof} 
Now, we need to show that $f_T$ is a valid ultra valuation function.

\begin{lemma}
\label{lem:ultra-validity}
For every $T\subseteq A$ with $|T|=r$, the function $f_T$ is a normalized, monotone, ultra reward function. 
\end{lemma}
\begin{proof}
We first verify the ultra definition. Fix $X,Y\subseteq A$ with $|X|\leq|Y|$ and $x\in X\setminus Y$. By \Cref{lem:ultra-single-dent}, there exists $y\in Y\setminus X$ such that 
\[
-\eps\mathbf{1}[X=T] - \eps\mathbf{1}[Y=T] \leq -\eps\mathbf{1}[X-x+y=T] - \eps\mathbf{1}[Y-y+x=T]. 
\]
An exchange preserves the cardinality of each set, and therefore 
\[
\phi(|X|)+\phi(|Y|) = \phi(|X-x+y|)+\phi(|Y-y+x|). 
\]
Adding these two relations gives 
\[
f_T(X)+f_T(Y) \leq f_T(X-x+y)+f_T(Y-y+x), 
\]
so $f_T$ is ultra. Normalization follows from $f_T(\emptyset)=0$. To prove monotonicity, it suffices to consider sets that differ by one element. If neither set is $T$, the claim follows because $\phi$ is nondecreasing. The only exceptional transitions are those entering and leaving $T$. For every $i\in T$ and $x\notin T$, 
\[
f_T(T\setminus\{i\}) = \eps(r-1) = f_T(T), \qquad (f_T)_x(T\cup\{x\}) = \frac12+\eps>0. 
\]
Thus, $f_T$ is monotone. Finally, normalization and monotonicity imply that $f_T$ is nonnegative. Moreover,
\[
f_T(S) \leq \phi(|S|) \leq \frac12+\eps r = \frac{3k-1}{4k} <1.
\]
Hence, $f_T:2^A\rightarrow[0,1]$. 
\end{proof}

\subsection{The Optimal Sets} 
\label{subsec:ultra-optimal-sets}
Give every agent the same cost 
$
c_i:=\frac{1}{2k}. 
$ 
For every hidden set $T$, let $g_T$ denote the principal's utility function induced by the reward function $f_T$ and the costs $c_i$. 
\begin{lemma}[Characterization of profitable sets]
\label{lem:ultra-profitable-sets}
For every hidden set $T$, a set $S\subseteq A$ has positive principal utility if and only if 
\[ 
S=T\cup\{x\} 
\] 
for some $x\notin T$. Every such set is optimal and has utility 
\begin{equation} 
\label{eq:ultra-optimum} 
G := \frac{3k-1}{4k^2(2k+1)} = \Theta(n^{-2}). \end{equation} 
Thus, $G$ is the optimal principal utility in every constructed instance. 
\end{lemma} 

\begin{proof}
We consider the possible cardinalities of $S$. The empty set has utility zero. Suppose first that $1\leq|S|<k$. If $S\neq T$, then every agent $i\in S$ has marginal contribution $\eps$. Consequently, 
\[ 
\sum_{i\in S} \frac{c_i}{(f_T)_i(S)} = |S|\frac{1/(2k)}{1/(4k)} = 2|S| > 1,
\] 
and hence $g_T(S)<0$. If $S=T$, every selected agent has zero marginal contribution, so $g_T(T)=-\infty$. Now suppose that $|S|=k$. Every such set has value $\frac12+\eps r$. Moreover, for every $i\in S$, 
\begin{equation} 
\label{eq:ultra-k-marginals} 
(f_T)_i(S) = \frac12 + \eps\mathbf{1}[S\setminus\{i\}=T]. 
\end{equation}
Indeed, deleting $i$ gives an $r$-set of value $\eps r$, unless the resulting set is $T$, in which case its value is $\eps(r-1)$. If $T\nsubseteq S$, then no deletion of $S$ equals $T$, so every marginal contribution in \eqref{eq:ultra-k-marginals} is $1/2$. Therefore, \[ 
\sum_{i\in S} \frac{c_i}{(f_T)_i(S)} = k\frac{1/(2k)}{1/2} = 1, 
\] 
and $g_T(S)=0$. It remains to consider $S=T\cup\{x\}$ for some $x\notin T$. Every agent in $T$ has marginal contribution $1/2$, whereas the marginal contribution of $x$ is $1/2+\eps$. Hence, 
\[
\begin{aligned}
\sum_{i\in S} \frac{c_i}{(f_T)_i(S)} &= \frac{k-1}{k} + \frac{1/(2k)}{1/2+1/(4k)}\\ &= \frac{k-1}{k} + \frac{2}{2k+1} = 1-\frac{1}{k(2k+1)}.
\end{aligned} 
\]
Since 
\[
f_T(S) = \frac12+\eps r = \frac{3k-1}{4k}, 
\] 
we obtain 
\[ 
\begin{aligned}
g_T(S) &= \frac{1}{k(2k+1)} \cdot \frac{3k-1}{4k}= \frac{3k-1}{4k^2(2k+1)} = G > 0. \end{aligned} 
\] 
Finally, suppose that $|S|>k$. Deleting any element changes neither term in the cardinality profile $\phi$, and neither $S$ nor any of its one-element deletions can equal the $r$-set $T$. Thus, every selected agent has zero marginal contribution, and $g_T(S)=-\infty$. We conclude that the only sets with positive principal utility are the sets $T\cup\{x\}$, and all of them have utility $G$. Every other set has nonpositive utility, so these sets are optimal. \end{proof} 
\subsection{Negligible Probability of Identifying the Hidden Set} \label{subsec:ultra-value-lower-bound}
By the definition of $f_T$, 
\begin{equation} 
\label{eq:ultra-cardinality-answers} f_T(S)=\phi(|S|) \quad\text{for every }S\neq T, \qquad f_T(T)=\phi(r)-\eps. 
\end{equation} 
Thus, unless the queried set is exactly $T$, the answer to a value query is determined solely by the cardinality of the queried set. Let 
$
\mathcal{T} := \{T\subseteq A:|T|=r\}, 
$
and in particular, $|\mathcal{T}|=\binom{n}{r}$. 
\begin{lemma}[Value-query bound] 
\label{lem:ultra-query-bound} Let $\mathcal{A}$ be a deterministic algorithm that makes at most $q$ value queries. Suppose that $T$ is chosen uniformly from $\mathcal{T}$, and let $\widehat{S}_T$ be the set returned by $\mathcal{A}$ on the instance induced by $f_T$. Then 
\[ 
\Pr_T\!\left[g_T(\widehat{S}_T)>0\right] \leq \frac{q+k}{\binom{n}{r}}. 
\]
Consequently, 
\[ 
\mathbb{E}_T\!\left[g_T(\widehat{S}_T)\right] \leq \frac{q+k}{\binom{n}{r}}\,G. 
\] 
\end{lemma} 
\begin{proof}
We may assume that, after computing its output $\widehat{S}$ but before returning it, the algorithm queries every set $\widehat{S}\setminus\{i\}$, $i\in\widehat{S}$, whenever $|\widehat{S}|=k$, and ignores the answers to these additional queries. This modification does not change the returned set and adds at most $k$ value queries. By \Cref{lem:ultra-profitable-sets}, if $g_T(\widehat{S})>0$, then $\widehat{S}=T\cup\{x\}$ for some $x\notin T$. In this case, one of the additional queries is exactly $T$. Thus, positive utility implies that the modified algorithm queries the hidden set. Consider now the execution in which every value query $S$ is answered with $\phi(|S|)$. Since the algorithm is deterministic, these answers fix its entire adaptive execution, including the additional queries described above. Let $\mathcal{Q}\subseteq\mathcal{T}$ be the collection of $r$-sets queried during this execution. The total number of queries is at most $q+k$, and therefore 
\[
|\mathcal{Q}|\leq q+k. 
\]
For every $T\notin\mathcal{Q}$, the execution on the instance induced by $f_T$ is identical to this cardinality-based execution. Indeed, suppose inductively that all previous answers coincide. The algorithm then makes the same next query $S$. Since $T\notin\mathcal{Q}$, we have $S\neq T$, and \eqref{eq:ultra-cardinality-answers} implies that the answer is $f_T(S)=\phi(|S|)$. Thus, the next answers coincide as well. It follows that, for every $T\notin\mathcal{Q}$, the modified algorithm never queries $T$ and therefore cannot return a set of positive utility. Since $T$ is uniform in $\mathcal{T}$, 
\[ 
\Pr_T\!\left[g_T(\widehat{S}_T)>0\right] \leq \frac{|\mathcal{Q}|}{|\mathcal{T}|} \leq \frac{q+k}{\binom{n}{r}}. 
\]
By \Cref{lem:ultra-profitable-sets}, every positive-utility output has utility exactly $G$, whereas every other output has utility at most zero. Therefore,
\[
\begin{aligned} \mathbb{E}_T\!\left[g_T(\widehat{S}_T)\right] &\leq \Pr_T\!\left[g_T(\widehat{S}_T)>0\right]G\\ &\leq \frac{q+k}{\binom{n}{r}}\,G. \end{aligned}
\]
\end{proof} 
\subsection{Putting It Together}
We are now ready to prove \Cref{thm:ultra-hardness}. 
\begin{proof}[Proof of Theorem~\ref{thm:ultra-hardness}] 
Choose $T$ uniformly from $\mathcal{T}$. By \Cref{lem:ultra-query-bound}, every deterministic algorithm making at most $q$ value queries obtains expected utility at most \[ 
\frac{q+k}{\binom{n}{r}}\,G
\]
under this distribution. By Yao's principle, for every randomized algorithm making at most $q$ value queries, there exists a fixed $T\in\mathcal{T}$ for which 
\begin{equation} 
\label{eq:ultra-randomized-bound} 
\frac{\mathbb{E}[g_T(\widehat{S})]}{G} \leq \frac{q+k}{\binom{n}{r}}, \end{equation} where the expectation is over the internal randomness of the algorithm. Indeed, this follows by conditioning on the algorithm's random seed, applying \Cref{lem:ultra-query-bound} to the resulting deterministic algorithm, and then averaging over the seeds and over $T$. Since $r=\lfloor n/2\rfloor$, $k\leq n$, and 
\[ 
\binom{n}{r} = \binom{n}{\lfloor n/2\rfloor} \geq \frac{2^n}{n+1}, 
\] 
for every polynomial $q=q(n)$, 
\[ 
\frac{q+k}{\binom{n}{r}} \leq \frac{(n+1)(q+n)}{2^n} = 2^{-\Omega(n)}. 
\] 
Thus, on some instance, every randomized algorithm using polynomially many value queries obtains only a $2^{-\Omega(n)}$ fraction of the optimal utility in expectation. A $2^{o(n)}$-approximation would have to obtain a $2^{-o(n)}$ fraction of the optimum, contradicting \eqref{eq:ultra-randomized-bound}. 
\end{proof}

%% file: appendix.tex
\appendix
\input{no-ptas-appendix}
\input{wmrf_appendix}
\input{better_const_gs_2}
\input{better_const_SM}

%% file: no-ptas-appendix.tex
\section{Deferred Proofs from \Cref{sec:no-ptas-gs}}
\label{app:oxs-proofs}
\subsection{Utility of Assignment Sets}
\label{app:ass-set-ut}
\assignmentsetutilitylemma*

\begin{proof}[Proof of \Cref{lem:assignment-set-utility}]
Match every selected state agent $u^{\rho(u)}$ to its vertex slot $s_u$,
and every detector $d_C$ to its clause slot $s_C$. This matching has value
$(N+3m)/Z=1$ and is optimal, since the total weight available from all slots
is at most $1$. Hence, $f(S_\rho)=1$.

Removing a state agent $u^{\rho(u)}$ leaves its vertex slot unused, since
the other state agent of $u$ is not selected and all clause slots remain
occupied by their detectors. Thus, its marginal contribution is $1/Z$, and
its required payment share is
\[
    \frac{c_{u^{\rho(u)}}}
         {f_{u^{\rho(u)}}(S_\rho)}
    =
    \frac{\eta/(100mZ)}{1/Z}
    =
    \frac{\eta}{100m}.
\]
Now consider a detector $d_C$. If $C$ is not satisfied by $\rho$, no
selected state agent has a weight-$2/Z$ edge to $s_C$. The detector's
marginal contribution is therefore $3/Z$, and its required payment share is
\[
    \frac{c_{d_C}}
         {f_{d_C}(S_\rho)}
    =
    \frac{\eta/(mZ)}{3/Z}
    =
    \frac{\eta}{3m}.
\]

If $C$ is satisfied, a satisfying state agent can move from its
weight-$1/Z$ vertex edge to the newly available weight-$2/Z$ clause edge.
This recovers exactly $1/Z$, since the vertex slot it leaves cannot be used
by another selected agent. Hence, the detector's marginal contribution is
$2/Z$, and its required payment share is
\[
    \frac{c_{d_C}}
         {f_{d_C}(S_\rho)}
    =
    \frac{\eta/(mZ)}{2/Z}
    =
    \frac{\eta}{2m}.
\]
In particular, every selected agent has positive marginal contribution, so
$S_\rho$ is incentivizable.
Using
$\operatorname{sat}(\rho)=m-\operatorname{cut}(\rho)$ and $N=m/3$, we get
\[
    t(S_\rho)
    =
    \frac{N\eta}{100m}
    +\operatorname{sat}(\rho)\frac{\eta}{2m}
    +\bigl(m-\operatorname{sat}(\rho)\bigr)\frac{\eta}{3m}
    =
    \frac{\eta}{300}
    +\frac{\eta}{2}
    -\frac{\eta}{6m}\operatorname{cut}(\rho).
\]
Substituting $f(S_\rho)=1$, we get
    $g(S_\rho)
    =
    1-t(S_\rho)
    =
    1-\frac{\eta}{300}
    -\frac{\eta}{2}
    +\frac{\eta}{6m}\operatorname{cut}(\rho)$, as required.
\end{proof}

\subsection{Assignment Sets Are Without Loss of Generality}
\label{app:ass-set-wlog}
  \assignmentsetlemma*

  \begin{proof}[Proof of \Cref{lem:assignment-set}] 
Fix an incentivizable set $S$, and let $M$ be a maximum-weight matching
between the agents in $S$ and the slots. Since every agent has strictly
positive cost, every selected agent has positive marginal contribution and
must therefore be matched in $M$.
Let $d$ be the number of selected detector agents. Every selected detector
$d_C$ is matched to its own slot $s_C$: if $s_C$ were empty, adding its
weight-$3/Z$ edge would increase the matching value, while if it were
occupied by a state agent along a weight-$2/Z$ edge, replacing that agent
by $d_C$ would also increase the value.

Consequently, every selected state agent is matched either to a vertex slot
by a weight-$1/Z$ edge or to a clause slot whose detector is absent by a
weight-$2/Z$ edge. Let $k$ be the number of state agents matched to vertex
slots. Since at most $m-d$ state agents can be matched to clause slots,
\[
    f(S)
    \leq
    \frac{3d+k+2(m-d)}{Z}
    =
    1-\frac{(m-d)+(N-k)}{Z}.
\]
Set
\[
    \ell:=(m-d)+(N-k).
\]
Then
\[
    1-f(S)\geq\frac{\ell}{Z}.
\]

Every selected detector has marginal contribution either $2/Z$ or $3/Z$.
Indeed, removing a detector opens one additional clause slot, and since all
selected state agents were already matched, rematching can recover at most
$1/Z$. Similarly, every selected state agent has marginal contribution
either $1/Z$ or $2/Z$. Moreover, a state agent matched to a vertex slot has
marginal exactly $1/Z$.
It follows that
\[
    t(S)
    \leq
    m\frac{\eta}{2m}
    +2N\frac{\eta}{100m}
    =
    \frac{\eta}{2}+\frac{\eta}{150}
    <\eta.
\]
The same bound holds for every assignment set.

We construct an assignment $\rho$ from $M$. If one of $u^0,u^1$ is matched
to $s_u$, we choose its state as $\rho(u)$. Otherwise, if at least one of
$u^0,u^1$ belongs to $S$, we choose one such state; if neither belongs to
$S$, we choose $\rho(u)$ arbitrarily. This assignment is computable in
polynomial time.

We now bound the possible increase from $t(S)$ to $t(S_\rho)$. First,
$S_\rho$ adds the $m-d$ missing detectors, each requiring a share of at most
$\eta/(2m)$.

Second, a detector already selected in $S$ becomes more expensive only if
its marginal changes from $3/Z$ in $S$ to $2/Z$ in $S_\rho$, increasing
its payment by $\eta/(6m)$. If such a detector corresponds to a clause $C$,
then $C$ is satisfied by a state chosen by $\rho$. This state cannot have
been matched to its vertex slot in $M$, since after removing $d_C$ it could
then move to $s_C$, showing that the detector's original marginal was
already $2/Z$. Thus, every such detector can be associated with one of the
$N-k$ vertices whose slot is unoccupied in $M$. Since each chosen state
satisfies at most three clauses, at most $3(N-k)$ detector payments
increase.

Finally, if a vertex slot is occupied in $M$, $\rho$ retains the state agent
matched to it, whose payment is unchanged. For each of the remaining
$N-k$ vertices, the assignment set contains one state agent requiring a
share of at most $\eta/(100m)$. Therefore,
\[
\begin{aligned}
    t(S_\rho)-t(S)
    &\leq
    (m-d)\frac{\eta}{2m}
    +(N-k)\frac{\eta}{2m}
    +(N-k)\frac{\eta}{100m}
    \leq
    \frac{51\eta}{100m}\ell.
\end{aligned}
\]

\noindent Using $f(S_\rho)=1$, $t(S_\rho)<\eta$, and
$Z=N+3m=10m/3$, we conclude that
\[
\begin{aligned}
    g(S_\rho)-g(S)
    &=
    (1-f(S))(1-t(S_\rho))
    -f(S)\bigl(t(S_\rho)-t(S)\bigr)\\
    &\geq
    \frac{\ell}{Z}(1-\eta)
    -\frac{51\eta}{100m}\ell
    =
    \frac{30-81\eta}{100m}\ell
    \geq0,
\end{aligned}
\]
where the final inequality follows from our choice of $\eta$.
\end{proof}

%% file: wmrf_appendix.tex
\section{Deferred Proofs from \Cref{sec:wmrf}}
\label{app:wmrf-independent-proof}
\wmrfindependentoptimumlemma*
\begin{proof}[Proof of \Cref{lem:wmrf-independent-optimum}]
The claim is immediate if $g(S)=-\infty$. Otherwise, let
$T\subseteq S$ be a maximum-weight independent subset of $S$, so that
\[
    f(S)=f(T)=\sum_{i\in T}w_i.
\]
For every $i\in T$, we have
$T\setminus\{i\}\subseteq S\setminus\{i\}$. Since every WMRF is
submodular,
\[
    f_i(S)
    \leq
    f_i(T)
    =
    w_i,
\]
where the equality follows because $T$ is independent. Therefore,
\[
\begin{aligned}
    g(S)
    &=\left(
        1-\sum_{i\in S}
        \frac{c_i}{f_i(S)}
    \right)f(S)\\
    &\leq
    \left(
        1-\sum_{i\in T}
        \frac{c_i}{f_i(S)}
    \right)f(T)\\
    &\leq
    \left(
        1-\sum_{i\in T}\frac{c_i}{w_i}
    \right)
    \sum_{i\in T}w_i
    =
    g(T).
\end{aligned}
\]
Consequently, $    \max_{S\subseteq A}g(S)
    =
    \max_{S\in\mathcal I}g(S)$.
\end{proof}

%% file: better_const_gs_2.tex
\section{A Constant-Factor Approximation for Submodular Rewards}
\label{sec:gs-approximation}

In this section, we improve the approximation guarantee for submodular rewards with value- and demand-oracle access by refining the framework of~\cite{DEFK23}. Our analysis accounts for the payments in an optimal contract to guide the choice of which agents to consider and the prices used in demand queries.

Note that for GS rewards, exact demand queries can be implemented using polynomially many value queries, so the same guarantee holds with value-oracle access alone.
\begin{theorem}[Improved approximation for submodular rewards]
\label{thm:gs-improved-approximation}
For every $\eps\in(0,1)$, the optimal contract problem with a
submodular reward function admits a polynomial-time
$(3.287+\eps)$-approximation under value- and demand-oracle access.
\end{theorem}

Similar to the approach of \cite{DEFK23}, we split an optimal set $S^\star$ into \emph{light} agents, who can be incentivized to work alone for a relatively small payment and at most one \emph{heavy} agent, who requires a large payment to be incentivized to work. However, we differ in the threshold which determines which agents are light or heavy.  While \cite{DEFK23} fixed that threshold at $1/2$,  we optimize it as a fixed large fraction of the total payment needed to incentivize $S^\star$.

By considering heavy and light agents separately, the problem decomposes into taking either the best singleton, or approximating well the reward achieved by light agents alone. \cite{DEFK23} achieve this by a demand query at carefully calibrated prices over several ``guesses'' of the reward of $S^\star$. We further optimize these guesses, again by taking into account the total payment needed to incentivize $S^\star$ as well as its reward into our analysis. Notably, while the total payment needed to incentivize $S^\star$ is considered in our analysis, knowing it in advance is not required by our algorithm, which essentially ``guesses'' its possible values.

\paragraph{Light and Heavy Agents.}For every $i\in A$, write $f_i:=f_i(\{i\})=f(\{i\})$ for the singleton reward, and let 
\[
    I:=\max\left\{0,\max_{i\in A}g(\{i\})\right\}
      =\max\left\{0,\max_{i\in A}(f_i-c_i)\right\},
\]
be the best utility obtainable from a singleton or the empty set. By
submodularity, an agent with $f_i=0$ has zero marginal contribution to every
set and may be discarded.
Fix an optimal set $S^\star$, and suppose that $g(S^\star)>0$, since otherwise the empty set is optimal. Let
\[
    t:=\sum_{i\in S^\star}
    \frac{c_i}{f_i(S^\star)}.
\]
Then $g(S^\star)=(1-t)f(S^\star)$ and, in particular, $t<1$.

Let $\tau = 0.8442$ and define
\[
    A_\tau:=\left\{i\in A:\frac{c_i}{f_i}\leq\tau t\right\},
    \qquad L:=S^\star\cap A_\tau,
    \qquad
    t_L:=\sum_{i\in L}
    \frac{c_i}{f_i(S^\star)}.
\]
We call the agents in $A_\tau$ {\em light} and the remaining agents {\em heavy}.

The set $L$ is thus the set of light agents in $S^\star$ and $t_L$ is the total payment share of the light optimal agents in the optimal contract.

\begin{observation}
\label{obs:gs-payment-decomposition}
At most one agent in $S^\star$ lies outside $L$.
\end{observation}

\begin{proof}
  By submodularity, incentivizing an agent as part of $S^\star$ requires at least as much payment as incentivizing him alone. Thus, every agent in $S^\star\setminus L$ requires a payment greater than $\tau t$. Since $\tau\geq 1/2$, two such agents would together require more than the total payment $t$, which is a contradiction.
\end{proof}

\paragraph{$x$-candidate set.} We next construct a candidate set of agents to incentivize and bound its utility in terms of $f(L)$.
For a scale $x>0$, assign each $i\in A_\tau$ the price $p_i:=\sqrt{c_i x}$ and compute a set\footnote{This can be done via a demand query by setting $p_i=2$ for any $i\notin A_\tau$.}
\[
    D(A_\tau,x)\in\arg\max_{S\subseteq A_\tau}
    \left\{f(S)-\sum_{i\in S}p_i\right\}.
\]
Starting from $D(A_\tau,x)$, delete its agents one at a time in an arbitrary order until reaching the empty set. An \emph{$x$-candidate set} is a set in this sequence whose reward is closest to $x/2$:
\begin{equation}
\label{eq:gs-candidate-definition}
    U(A_\tau,x)\in\arg\min_{S\in\text{deletion chain of $D(A_\tau,x)$}}
        \left|f(S)-\frac{x}{2}\right|.
\end{equation}
The following lemma shows that the utility of an $x$-candidate set is at least
$x/4$ minus some loss that depends on $I$. 
To bound this loss, we use the cutoff defining $A_\tau$ to control the largest singleton reward in the demand set.

\begin{lemma}[Utility of the $x$-candidate set]
\label{lem:gs-calibrated-demand}
If $x>0$ satisfies
\begin{equation}
\label{eq:gs-reachability}
    f(L)-\sqrt{t_L f(L)x}\geq\frac{x}{2},
\end{equation}
then
\[
    g(U(A_\tau,x))\geq\frac{x}{4}-\frac{I^2}{4x(1-\tau t)^2}.
\]
\end{lemma}
\begin{proof}
Write $D:=D(A_\tau,x)$.
Since $S^\star$ has a positive utility, every positive-cost agent in
$S^\star$ has a positive marginal contribution. Cauchy--Schwarz and submodularity give
\begin{equation}
\label{eq:sqpr-bound}
    \sum_{i\in L}\sqrt{c_i}
    =\sum_{i\in L}\sqrt{\frac{c_i}{f_i(S^\star)}}\sqrt{f_i(S^\star)}
    \leq\sqrt{\sum_{i\in L}\frac{c_i}{f_i(S^\star)}}
        \sqrt{\sum_{i\in L}f_i(S^\star)}
    =\sqrt{t_L}\sqrt{\sum_{i\in L}f_i(S^\star)}
    \leq\sqrt{t_L f(L)}.
\end{equation}
For the last inequality, order the agents in $L$ arbitrarily. Each
$f_i(S^\star)$ is at most the marginal of $i$ when it is added after its
predecessors in this order, and these latter marginals telescope to $f(L)$.
Since $L\subseteq A_\tau$, demand optimality gives
$f(D)-\sum_{i\in D}p_i\geq f(L)-\sum_{i\in L}p_i$. We get
\[
    f(D)\geq f(D)-\sum_{i\in D}p_i
      \geq f(L)-\sum_{i\in L}p_i = f(L)-\sum_{i\in L}\sqrt{c_i x}
      \geq f(L)-\sqrt{t_L f(L)x},
\]
where the last inequality follows from \Cref{eq:sqpr-bound}. Together with \eqref{eq:gs-reachability}, this gives $f(D)\geq x/2$.
We now bound the total payment required to incentivize $U(A_\tau,x)$. Demand optimality and submodularity give $f_i(U(A_\tau,x)) \ge f_i(D) \ge p_i$ for all $i\in U(A_\tau, x)$. Thus,
\[
\begin{aligned}
\sum_{i\in U(A_\tau,x)} \frac{c_i}{f_i(U(A_\tau,x))} &\le \sum_{i\in U(A_\tau,x)} \frac{c_i}{p_i} = \sum_{i\in U(A_\tau,x)} \frac{1}{x}p_i\\ &\le \frac{1}{x}\sum_{i\in U(A_\tau, x)} f_i(U(A_\tau,x)) \\ &\le \frac{f(U(A_\tau,x))}{x},
\end{aligned}
\]
where the equality follows by substituting $p_i = \sqrt{c_i x}$.

We next note that by definition of $U(A_\tau,x)$, together with subadditivity, it holds that
\[
    \left|f(U(A_\tau,x))-\frac{x}{2}\right|
      \leq\frac12\max_{i\in D}f_i.
\]
In total we get
\[
\begin{split}
    g(U(A_\tau,x))&= \left(1-\sum_{i\in U(A_\tau,x)} \frac{c_i}{f_i(U(A_\tau,x))}\right) f(U(A_\tau,x)) \ge \left(1-\frac{f(U(A_\tau,x))}{x}\right) f(U(A_\tau,x)) \\
    &=\frac{x}{4}-\frac{(f(U(A_\tau,x))-x/2)^2}{x}
    \geq\frac{x}{4}-\frac{1}{4x}\left(\max_{i\in D}f_i\right)^2,
    \end{split}
\]
Every $i\in D$ belongs to $A_\tau$, so the definition of $I$ and the cutoff give
\[
    I\geq f_i-c_i\geq(1-\tau t)f_i.
\]
Since $1-\tau t>0$, we obtain $\max_{i\in D}f_i\leq I/(1-\tau t)$.
Substituting this into the preceding bound yields
\[
    g(U(A_\tau,x))\geq\frac{x}{4}
    -\frac{I^2}{4x(1-\tau t)^2},
\]
proving the claim.
\end{proof}

\paragraph{The approximation guarantee.}
We now use the utility bound in \Cref{lem:gs-calibrated-demand} to obtain
an approximation to $g(S^\star)$.
We do so by choosing a suitable value of $x$. Notably, the choice of $\tau=0.8442$ was used to optimize the approximation ratio with respect to our analysis.

\begin{proposition}[Approximation guarantee]
\label{prop:gs-demand-approximation}
If $f(L)>0$, define
\[
    x^\star:=\lambda(t_L)f(L),
    \qquad
    \text{where }\lambda(r):=(\sqrt{r+2}-\sqrt r)^2.
\]
For every $0<x\le x^\star$,
\[
    \max\{I,g(U(A_\tau,x))\}
    \ge
    \frac{x}{x^\star}\frac{1000}{3287}\,g(S^\star).
\]
If $f(L)=0$, the best singleton is optimal.
\end{proposition}

\begin{proof}
If $f(L)=0$, \Cref{obs:gs-payment-decomposition} and submodularity imply that the
best singleton is optimal. Otherwise, fix $x=\rho x^\star$,
where $0<\rho\le1$, and write
\[
    M:=\max\{I,g(U(A_\tau,x))\},
    \qquad
    u:=1-\tau t.
\]
Note that $x^\star$ exactly satisfies \Cref{eq:gs-reachability} as an equality. Since $x\le x^\star$, \Cref{eq:gs-reachability} holds for $x$ as well. Since $I\leq M$ and $u=1-\tau t$,
     \Cref{lem:gs-calibrated-demand} gives
\[
    M\geq g(U(A_\tau,x)) \geq\frac{x}{4}-\frac{I^2}{4xu^2}
    \geq\frac{x}{4}-\frac{M^2}{4xu^2}.
\]
Solving this quadratic inequality for $x$, and writing
$h(u):=\sqrt{1+4u^2}+2u$, yields
$x\le \frac{h(u)}{u}M$, and by substituting $f(L)=x^{\star}/\lambda(t_L)$ and $x=\rho x^\star$, we get
\[
    f(L)\le \frac{h(u)}{\rho\lambda(t_L)u}M.
\]
If $S^\star=L$, then $t_L=t$, so
\[
    g(S^\star)\le \frac{R_0(t)}{\rho}M,
    \quad
    \text{where }
    R_0(t):=
    \frac{(1-t)h(1-\tau t)}{\lambda(t)(1-\tau t)}.
\]
Otherwise, let $j$ be the unique heavy optimal agent.
Its payment share is $t-t_L>\tau t$, and
\[
    I\ge (1-t+t_L)f_j.
\]
Consequently, subadditivity and the bound on $f(L)$ give
\[
\begin{aligned}
    g(S^\star)
    &\le (1-t)f(L)+\frac{1-t}{1-t+t_L}I\\
    &\le \frac{M}{\rho}
    \left(
        \frac{(1-t)h(1-\tau t)}
             {\lambda(t_L)(1-\tau t)}
        +\frac{1-t}{1-t+t_L}
    \right),
\end{aligned}
\]
where the last inequality uses $I\le M$ and $\rho\le1$.
For a fixed $t_L$, the expression in parentheses is nonincreasing
in $t$: the positive factors $(1-t)/(1-\tau t)$ and
$h(1-\tau t)$ decrease, and the last term is nonincreasing.
Since $t>t_L/(1-\tau)$, it is bounded by its value at
$t=t_L/(1-\tau)$. Thus
\[
    g(S^\star)\le
    \frac{M}{\rho}\sup_{0\le z<1}R_1(z),
    \quad
    \text{where }
    R_1(z):=
    \frac{1-z}{1-\tau z}
    \left(
        1+\frac{h(1-\tau z)}{\lambda((1-\tau)z)}
    \right).
\]
For $\tau=0.8442$, the two functions satisfy
\[
    \sup_{0\le t<1}R_0(t)<3.286883,
    \qquad
    \sup_{0\le t<1}R_1(t)<3.286886,
\]
which proves the proposition.
\end{proof}

\paragraph{Finding a good candidate.}
The set $A_\tau$ and the scale $x^\star$ both depend on the unknown optimal contract $S^\star$.
The set $A_{\tau}$ is one of the prefixes obtained by
sorting agents by $c_i/f_i$ and grouping equal ratios, and is therefore easily enumerated over.
Since $A_\tau$ is unknown, for each enumerated prefix $P$ and each scale
$x>0$, we make the demand query at prices $p_i=\sqrt{c_i x}$ over subsets
of $P$ and denote its output by $D(P,x)$. We then apply the deletion-chain
rule in \eqref{eq:gs-candidate-definition} to $D(P,x)$ to obtain the
candidate $U(P,x)$. To approximate $x^\star$, we use a geometric grid,
as described in the following lemma.

\begin{lemma}[Discretizing the demand scale]
\label{lem:gs-scale-grid}
Suppose $I>0$. Fix $\delta\in(0,1)$ and set
\[
    K:=\left\lceil\log_{1+\delta}(27n)\right\rceil,
    \qquad x_\ell:=\frac{I}{2}(1+\delta)^\ell
    \quad(\ell=0,\ldots,K).
\]
Then there exists some
$\ell\in\{0,\ldots,K\}$ such that
\[
    \max\left\{I,g\bigl(U(A_{\tau},x_\ell)\bigr)\right\}
    \geq\frac{1000}{3287(1+\delta)}\,g(S^\star).
\]
\end{lemma}

\begin{proof}
If $f(L)\le I$, then $g(S^\star)\le2I$: this is immediate
when $S^\star=L$; otherwise, the unique heavy optimal agent $j$
satisfies $(1-t)f_j\le I$, and hence
\[
    g(S^\star)
    \le (1-t)(f(L)+f_j)
    \le f(L)+I
    \le 2I.
\]
Thus the claimed guarantee follows from the singleton alone.

Now suppose $f(L)>I$. Since $1/2<\lambda(t_L)\le2$, and every
$i\in L$ satisfies $(1-\tau)f_i\le I$, subadditivity gives
\[
    \frac{I}{2}<x^\star
    \le 2f(L)
    \le \frac{2nI}{1-\tau}
    < \frac{27nI}{2}
    \le x_K.
\]
The largest grid point $x_\ell\le x^\star$ therefore satisfies
$\frac{x^\star}{1+\delta}\le x_\ell\le x^\star$.
Applying Proposition~\ref{prop:gs-demand-approximation}
at $x=x_\ell$ gives
\[
    \max\{I,g(U(A_\tau,x_\ell))\}
    \ge \frac{1000}{3287(1+\delta)}\,g(S^\star),
\]
as required.
\end{proof}

\Cref{alg:gs-approximation}  goes over all options of $A_{\tau}$ and values of $x$ prescribed by \Cref{lem:gs-scale-grid} and
returns a set of maximum utility among the candidate sets, the
singletons, and the empty set. Thus, it does not require the unknown
values $t$, $t_L$, or $f(L)$.

\begin{algorithm}[H]
\DontPrintSemicolon
\caption{Approximation for the optimal contract under submodular $f$}
\label{alg:gs-approximation}
\KwIn{A submodular reward function $f:2^A\rightarrow[0,1]$ given by
value and demand oracles, costs $\{c_i\}_{i\in A}$, and $\eps\in(0,1)$}
\KwOut{A set of agents $\widehat S$ to incentivize}
Query $f_i:=f(\{i\})$ for every $i\in A$, and discard agents with $f_i=0$\;
Compute $I$, and initialize $\widehat S$ to a singleton or the empty set
with utility $I$\;
\If{$I=0$}{\Return{$\emptyset$}\;}
Set $\delta:=\eps/4$ and
$K:=\left\lceil\log_{1+\delta}(27n)\right\rceil$\;
Sort agents by nondecreasing $c_i/f_i$,  and let
$\mathcal P$ be the set of all prefixes \;
\ForEach{$P\in\mathcal P$}{
    \For{$\ell=0,\ldots,K$}{
        Set $x_\ell:=\frac{I}{2}(1+\delta)^\ell$\;
        Compute the $x_\ell$-candidate $U(P,x_\ell)$ from
        \eqref{eq:gs-candidate-definition}\;
        Set $\widehat S$ to a set of maximum utility in
        $\{\widehat S,U(P,x_\ell)\}$\;
    }
}
\Return{$\widehat S$}\;
\end{algorithm}

\begin{proof}[Proof of \Cref{thm:gs-improved-approximation}]

The prefix $A_{\tau}$ is among those considered by
\Cref{alg:gs-approximation}, and the appropriate $x_\ell$ as prescribed by \Cref{lem:gs-scale-grid} as well, thus either the best singleton or some $U(P, x_\ell)$ achieves a utility of at least $\frac{1000}{3287(1+\delta)}g(S^\star)$. The algorithm outputs the best set among these, so with the choice of $\delta=\eps/4$,
its approximation factor is at most
\[
    \frac{3287}{1000}\left(1+\frac{\eps}{4}\right)<3.287+\eps.
\]
Finally, there are at most $n+1$ prefixes and
$K+1=O(\eps^{-1}\log(n+1))$ considered values of $x$. Each candidate uses one demand query
and $O(n)$ additional value queries to evaluate the deletion chain and the
utility of its selected set. Hence the algorithm runs in polynomial time.
\end{proof}

%% file: better_const_SM.tex
\section{Extension to Submodular Rewards with Value Queries}
\label{sec:submodular-value-oracle}

To obtain an approximation using only value queries, we replace the exact
demand query in \Cref{sec:gs-approximation} by an approximate demand
computation. Pruning its output restores the marginal inequalities needed
for the $x$-candidate's utility bound.

\begin{theorem}[Approximation for submodular rewards]
\label{thm:submodular-improved-approximation}
For every $\eps\in(0,1)$, the optimal contract problem with a submodular
reward function admits a polynomial-time $(6.128+\eps)$-approximation under
value-oracle access.
\end{theorem}

We follow the construction and analysis of \Cref{sec:gs-approximation},
specifying the changes to the $x$-candidate, the cutoff, and the scale.
Once the corresponding utility bound is established, the singleton
comparison and the enumeration of prefixes and scales apply with these
new parameters.

\paragraph{Light and Heavy Agents.}
We use $f_i$, $I$, $S^\star$, and $t$ as in \Cref{sec:gs-approximation},
discarding agents with $f_i=0$ and assuming $g(S^\star)>0$ for the analysis.
In this section, fix $\tau=1/2$ and define
\[
    A_\tau:=\left\{i\in A:\frac{c_i}{f_i}\leq\tau t\right\},
    \qquad L:=S^\star\cap A_\tau.
\]
As before, $t_L$ is the total payment share of the agents in $L$ in the
optimal contract. The argument of \Cref{obs:gs-payment-decomposition}
applies at this cutoff as well, so at most one optimal agent is heavy.

\paragraph{$x$-candidate set.}
We adapt the construction of a candidate set of agents to incentivize by
using the following approximate demand guarantee, as in~\cite{DEFK23}.
Let $\beta:=1-1/e$.

\begin{lemma}[{Harshaw et al.~[2019]; Sviridenko et al.~[2017]}]
\label{lem:submodular-approximate-demand}
Let $f$ be a monotone submodular function given by a value oracle, and let
$\{p_i\}_{i\in A_\tau}$ be nonnegative prices. In polynomial time, one can
compute a set $\widetilde D\subseteq A_\tau$ satisfying
\begin{equation}
\label{eq:submodular-approximate-demand}
    f(\widetilde D)-\sum_{i\in\widetilde D}p_i
    \geq
    \beta f(T)-\sum_{i\in T}p_i
    \qquad\text{for every }T\subseteq A_\tau.
\end{equation}
\end{lemma}
We call $\widetilde D$ a $\beta$-approximate demand set. This guarantee
scales the reward of the comparison set by $\beta$, while subtracting its
full price. It does not ensure that every agent in $\widetilde D$ has
marginal contribution at least its price, which was needed for the
payment bound in \Cref{lem:gs-calibrated-demand}. We restore this
property by pruning the approximate demand set.

For a scale $x>0$, assign each $i\in A_\tau$ the price
$p_i:=\sqrt{c_i x}$ and obtain $\widetilde D$ from
\Cref{lem:submodular-approximate-demand}. Initialize $D:=\widetilde D$
and delete an agent $i$ whenever $f_i(D)<p_i$.
Write $D(A_\tau,x)$ for the set remaining when no such agent exists.
Thus, in this section, $D(A_\tau,x)$ denotes the pruned approximate demand set.
Starting from $D(A_\tau,x)$, form an arbitrary fixed deletion chain down
to $\emptyset$. An \emph{$x$-candidate set} $U(A_\tau,x)$ is a set in this
chain whose reward is closest to $x/2$, as in
\eqref{eq:gs-candidate-definition}.

The next lemma shows that this candidate satisfies the same utility bound
as in \Cref{lem:gs-calibrated-demand}, with the condition on the scale
adjusted for the factor $\beta$.
\begin{lemma}[Utility of the $x$-candidate set]
\label{lem:submodular-approximate-demand-candidate}
If $x>0$ satisfies
\begin{equation}
\label{eq:submodular-reachability}
    \beta f(L)-\sqrt{t_L f(L)x}\geq\frac{x}{2},
\end{equation}
then
\[
    g(U(A_\tau,x))\geq\frac{x}{4}
      -\frac{I^2}{4x(1-\tau t)^2}.
\]
\end{lemma}

\begin{proof}
Write $D:=D(A_\tau,x)$.
The Cauchy--Schwarz estimate \eqref{eq:sqpr-bound} is unchanged, so
$\sum_{i\in L}p_i\leq\sqrt{t_L f(L)x}$.
Deleting a violating agent increases reward minus prices by
$p_i-f_i(D)>0$. Consequently, pruning preserves the
approximate demand guarantee, and comparison with $L$ gives
\[
\begin{aligned}
    f(D)&\geq f(D)-\sum_{i\in D}p_i
       \geq f(\widetilde D)-\sum_{i\in\widetilde D}p_i\\
    &\geq\beta f(L)-\sum_{i\in L}p_i
       \geq\beta f(L)-\sqrt{t_L f(L)x}.
\end{aligned}
\]
At termination, $f_i(D)\geq p_i$ for every $i\in D$.
By submodularity, the same marginal bound holds for the agents in
$U(A_\tau,x)$. The payment calculation in the proof of
\Cref{lem:gs-calibrated-demand} therefore applies.
Under \eqref{eq:submodular-reachability}, the reward bound above gives
$f(D)\geq x/2$. The same deletion-chain argument then yields
\[
    g(U(A_\tau,x))\geq\frac{x}{4}
      -\frac{1}{4x}\left(\max_{i\in D}f_i\right)^2.
\]
Finally, every $i\in D\subseteq A_\tau$ satisfies
$I\geq f_i-c_i\geq(1-\tau t)f_i$.
Substituting $\max_{i\in D}f_i\leq I/(1-\tau t)$ concludes the proof.
\end{proof}

\paragraph{The approximation guarantee.}
We now use this utility bound in the approximation argument of
\Cref{prop:gs-demand-approximation}, with the cutoff $\tau=1/2$ and
the scale adjusted for the factor $\beta$.

\begin{proposition}[Approximation guarantee]
\label{prop:submodular-demand-approximation}
If $f(L)>0$, define
\begin{equation}
\label{eq:submodular-lambda}
    x^\star:=\lambda_\beta(t_L)f(L),
    \qquad\text{where}\quad
    \lambda_\beta(r):=\left(\sqrt{r+2\beta}-\sqrt r\right)^2.
\end{equation}
For every $0<x\leq x^\star$,
\[
    \max\{I,g(U(A_\tau,x))\}
      \geq\frac{x}{x^\star}\frac{125}{766}\,g(S^\star).
\]
If $f(L)=0$, then the best singleton is optimal.
\end{proposition}

\begin{proof}
The case $f(L)=0$ follows from the same argument as in the proof of
\Cref{prop:gs-demand-approximation}. Otherwise, fix $x=\rho x^\star$,
where $0<\rho\leq1$, and write $M:=\max\{I,g(U(A_\tau,x))\}$.
Our choice of $x^\star$ satisfies \eqref{eq:submodular-reachability}
with equality, so the condition also holds for $x\leq x^\star$.
Since $I\leq M$ and $\tau=1/2$,
\Cref{lem:submodular-approximate-demand-candidate} gives
\[
    M\geq\frac{x}{4}-\frac{M^2}{4x(1-t/2)^2}.
\]
The quadratic calculation in the proof of
\Cref{prop:gs-demand-approximation}, with $\lambda_\beta$ in place of
$\lambda$, gives
\[
    f(L)\leq
    \frac{h(1-t/2)}{\rho\lambda_\beta(t_L)(1-t/2)}M,
\]
where $h$ is the function defined in that proof.

The two case comparisons in the same proof apply with the cutoff
$\tau=1/2$ and $\lambda_\beta$ in place of $\lambda$.
If $S^\star=L$, then $t_L=t$ and
\begin{equation}
\label{eq:submodular-reciprocal-zero}
    g(S^\star)\leq\frac{M}{\rho}R_{0,\beta}(t),
    \quad\text{where}\quad R_{0,\beta}(t):=
    \frac{(1-t)h(1-t/2)}{\lambda_\beta(t)(1-t/2)}.
\end{equation}
Otherwise, $t_L<t/2$, and the comparison with the heavy optimal agent gives
\[
    g(S^\star)\leq\frac{M}{\rho}
    \left(
        \frac{(1-t)h(1-t/2)}{\lambda_\beta(t_L)(1-t/2)}
        +\frac{1-t}{1-t+t_L}
    \right).
\]
For a fixed $t_L$, the expression in parentheses decreases with $t$,
by the same monotonicity argument as in
\Cref{prop:gs-demand-approximation}. It is therefore bounded by its
value at $t=2t_L$, giving
\begin{equation}
\label{eq:submodular-reciprocal-one}
\begin{aligned}
    g(S^\star)&\leq\frac{M}{\rho}\sup_{0\leq z<1}R_{1,\beta}(z),
\quad \text{where}\quad R_{1,\beta}(z)&:=
      \frac{1-z}{1-z/2}
      \left(1+\frac{h(1-z/2)}{\lambda_\beta(z/2)}\right).
\end{aligned}
\end{equation}
For $\beta=1-1/e$, differentiation shows that each expression has a
unique maximum on $[0,1)$. Bounding these maxima gives
\[
    \sup_{0\leq t<1}R_{0,\beta}(t)<6.127981,
    \qquad \sup_{0\leq t<1}R_{1,\beta}(t)<5.848.
\]
Both bounds are below $6.128=766/125$, so
$M\geq\rho(125/766)g(S^\star)$, proving the claim.
\end{proof}

\paragraph{Finding a good candidate.}
The prefix $A_\tau$ is among the prefixes enumerated in
\Cref{alg:gs-approximation}. For each enumerated prefix $P$ and each
scale $x>0$, apply the approximate demand, pruning, and deletion-chain
construction above with $P$ in place of $A_\tau$, and denote the
resulting sets by $D(P,x)$ and $U(P,x)$.
To approximate the unknown scale $x^\star$,
we use a geometric grid with a different range. For $I>0$ and
$\delta\in(0,1)$, set
\[
    K:=\left\lceil\log_{1+\delta}(16n)\right\rceil,
    \qquad x_\ell:=\frac{I}{4}(1+\delta)^\ell
    \quad(\ell=0,\ldots,K).
\]
If $f(L)\leq I$, the argument at the beginning of the proof of
\Cref{lem:gs-scale-grid} gives $I\geq g(S^\star)/2$. The same argument
applies because the cutoff still leaves at most one heavy optimal agent.
Now suppose $f(L)>I$. Since $t_L<1$, $\beta>5/8$, and
$\lambda_\beta$ is decreasing,
\[
    \frac14<\lambda_\beta(1)<\lambda_\beta(t_L)
      \leq\lambda_\beta(0)=2\beta<2.
\]
For every $i\in L$, the cutoff gives
$I\geq(1-c_i/f_i)f_i\geq(1-t/2)f_i>f_i/2$, so $f_i<2I$.
Subadditivity and the choice of $x^\star$ therefore give
\[
    f(L)\leq\sum_{i\in L}f_i<2nI,
    \quad\text{and}\quad \frac{I}{4}<x^\star<4nI\leq x_K.
\]
As in \Cref{lem:gs-scale-grid}, the largest grid point not exceeding
$x^\star$ satisfies
\[
    \frac{x^\star}{1+\delta}\leq x_\ell\leq x^\star.
\]
Applying \Cref{prop:submodular-demand-approximation} at $x=x_\ell$ gives
\[
    \max\{I,g(U(A_\tau,x_\ell))\}
      \geq\frac{x_\ell}{x^\star}\frac{125}{766}\,g(S^\star)
      \geq\frac{125}{766(1+\delta)}\,g(S^\star).
\]

We use the outer routine of \Cref{alg:gs-approximation} with
$\delta:=\eps/7$, the grid above, and the candidate construction of this
section. The routine starts with a singleton or the empty set of utility
$I$ and returns the best set among it and all candidate sets.

\begin{proof}[Proof of \Cref{thm:submodular-improved-approximation}]
If $g(S^\star)=0$, the initial set of utility $I\geq0$ already
meets the guarantee. Otherwise, $I>0$ by submodularity.
The prefix and scale enumeration above then guarantees that the
returned set has utility at least
$125g(S^\star)/(766(1+\delta))$. With $\delta=\eps/7$, the approximation
factor is therefore at most
\[
    \frac{766}{125}\left(1+\frac{\eps}{7}\right)<6.128+\eps.
\]
There are at most $n+1$ prefixes and $O(\eps^{-1}\log(n+1))$ scales.
Each candidate uses one approximate demand computation, which requires
polynomially many value queries by
\Cref{lem:submodular-approximate-demand}. Pruning performs at most $n$
deletions and uses $O(n^2)$ additional value queries; the deletion chain
and candidate utility can also be evaluated with polynomially many value
queries. Thus, the algorithm runs in polynomial time.
\end{proof}